\pdfoutput=1
\documentclass[format=acmsmall, screen, table]{acmart} 

\usepackage[acm,ownbibstyle]{pamas}

\setcitestyle{authoryear}
\acmJournal{TEAC}
\startPage{1}
\fancypagestyle{firstpagestyle}{
	\fancyhf{}%
	\fancyfoot[L]{}%
}
\fancypagestyle{acmec}{
	\fancyhf{}%
	\fancyhead[C]{\footnotesize{\mbox{\shortauthors}}}
	\fancyhead[R]{\thepage}%
}
\renewcommand{\footnotetextcopyrightpermission}[1]{\thankses}

\usepackage{booktabs}  
\usepackage{enumitem}
\usepackage{graphicx}
\usepackage{hyperref}
\usepackage{mathtools} 
\usepackage{tikz}
\usepackage{nicefrac}
\usepackage{subcaption}

\usetikzlibrary{positioning,fit,shapes,calc}

\newtheorem{remark}[theorem]{Remark}
	
\newcommand{\set}[1]{\{#1\}}
	\newcommand{\pref}{\succsim\xspace}

\newcommand{\midd}{\mathbin{:}}
	\renewcommand{\ie}{that is, }

	\newcommand{\type}{\theta}
	\newcommand{\types}{{\Theta}}
	
	\newcommand{\baker}{\type_1}

	\newcommand{\fracc}[2]{\frac{|#1|}{|#2|}}

	\newcommand{\nbh}[1][]{
		\ifthenelse{\equal{#1}{}}{\nu}{\nu(#1)}
	}
	
	\newcommand{\cstr}[1][]{
		\ifthenelse{\equal{#1}{}}{\pi}{\cstr(#1)}
	}

	\usepackage{xspace}

\begin{document}

\title{Fractional Hedonic Games}

\author{Haris Aziz}
\affiliation{Data61, CSIRO and UNSW Australia}
\author{Florian Brandl}
\affiliation{Technical University of Munich}
\author{Felix Brandt}
\affiliation{Technical University of Munich}
\author{Paul Harrenstein}
\affiliation{University of Oxford}
\author{Martin Olsen}
\affiliation{Aarhus University}
\author{Dominik Peters}
\affiliation{University of Oxford}

\date{}

\begin{abstract}
	The work we present in this paper initiated the formal study of \emph{fractional hedonic games}, coalition formation games in which the utility of a player is the average value he ascribes to the members of his coalition. 
	Among other settings, this covers situations in which players only distinguish between friends and non-friends and desire to be in a coalition in which the fraction of friends is maximal.
	Fractional hedonic games thus not only constitute a natural class of succinctly representable coalition formation games, but also provide an interesting framework for network clustering. We propose a number of conditions under which the core of fractional hedonic games is non-empty and provide algorithms for computing a core stable outcome. By contrast, we show that the core may be empty in other cases, and that it is computationally hard in general to decide non-emptiness of the core.
\end{abstract}

\maketitle

\section{Introduction}

Hedonic games present a natural and versatile framework to study the formal aspects of coalition formation which has received much attention from both an economic and an algorithmic perspective. This work was initiated by   \citet{DrGr80a}, \citet{BKS01a}, \citet{CeRo01a}, and \citet{BoJa02a} and has sparked a lot of follow-up work. A recent survey was provided by \citet{AzSa15a}.
In hedonic games, coalition formation is approached from a game-theoretic angle. The outcomes are coalition structures---partitions of the players---over which the players have preferences. Moreover, the players have different individual or joint strategies at their disposal to affect the coalition structure to be formed. Various solution concepts---such as the \emph{core}, the \emph{strict core}, and several kinds of \emph{individual stability}---have been proposed to analyze these games.

The characteristic feature of hedonic games is a non-ex\-ter\-nal\-ities condition, according to which every player's preferences over the coalition structures are fully determined by the player's preferences over coalitions he belongs to, and do not depend on how the remaining players are grouped. Nevertheless, the number of coalitions a player can be a member of is exponential in the total number of players, and the development and analysis of concise representations as well as interesting subclasses of hedonic games are an ongoing concern in computer science and game theory. Particularly prominent in this respect are representations in which the players are assumed to entertain preferences over the other players, which are then systematically lifted to preferences over coalitions~\citep[see, \eg][]{BKS01a,CeRo01a,BoJa02a,AlRe04a,DBHS06a,ABH11c}.

The work presented in this paper pertains to what we will call \emph{fractional hedonic games}, a subclass of hedonic games in which each player is assumed to have cardinal utilities or \emph{valuations} over the other players. These induce preferences over coalitions by considering the \emph{average valuation} of the members in each coalition. The higher this value, the more preferred the respective coalition is. 
Previously the min, max, and sum operators have been used respectively for \emph{hedonic games based on worst players}~\citep{CeRo01a}, \emph{hedonic games based on best players}~\citep{CeHa02a}, and \emph{additively separable hedonic games}~\citep{BKS01a}. 
Despite the natural appeal of taking the average value, fractional hedonic games have enjoyed surprisingly little attention prior to this work.\footnote{\citet{Hajd06a} first mentioned the possibility of using the average value in hedonic games but did not further analyze it.} 
Fractional hedonic games can be represented by a weighted directed graph where the weight of edge~$(i,j)$ denotes the value player~$i$ has for player~$j$. However, we will be particularly interested in games that can be represented by  \emph{un}directed and \emph{un}weighted graphs. Thus, such games have symmetric valuations that only take the values $0$ and $1$. With the natural graphical representation of these games, desirable outcomes for fractional hedonic games also provide an interesting angle on network clustering.

Many natural economic problems can be modeled as fractional hedonic games.
A particular economic problem that we will consider is what we refer to as \emph{Bakers and Millers}.
Suppose there are two types of players, bakers and millers, where individuals of the same type are competitors, trading with players of the other type.
Both types of players can freely choose the `neighborhood' in which to set up their enterprises; in the formal model, each neighborhood forms a coalition. Millers want to be situated in a neighborhood with as many purchasing bakers relative to competing millers as possible, so as to achieve a high price for the wheat they produce. On the other hand, bakers seek as high a ratio of the number of millers to the number of bakers as possible, so as to keep the price of wheat low and that of bread up. We show that these problems (which belong to the class of fractional hedonic games) always admit a core stable partition. This result generalizes to situations in which there are more than two types of player who want to keep the fraction of players of their own type as low as possible. Our study of the Baker and Millers setting is inspired by \citeauthor{Sche71a}'s famous dynamic model of segregation \citep{Sche71a,Sche78a}.

Another example concerns the formation of political parties. The valuation of two players for each other may be interpreted as the extent to which their opinions overlap, perhaps measured by the inverse of their distance in the political spectrum. In political environments, players need to form coalitions and join parties to acquire influence. On the other hand, as parties become larger, disagreement among their members will increase, making them susceptible to split-offs. Thus, one could assume that players seek to maximize the \emph{average} agreement with the members of their coalition.

The contributions of the paper are as follows.
\begin{itemize}
	\item We introduce and formally define fractional hedonic games and their graphical representation. We identify the subclass of games represented by undirected and unweighted graphs (\emph{simple} and \emph{symmetric} fractional hedonic games) and discuss some of their properties.
	\item We show that fractional hedonic games may have an empty core, even in the simple symmetric case. We give an example of such a game with $40$ players. We leverage this example to show that it is $\Sigma_2^p$-complete to decide whether a given simple symmetric fractional hedonic game has non-empty core. Thus, finding a partition in the core is NP-hard. It is also coNP-complete to verify whether a given partition is in the core.
	\item Based on the graphical representation of fractional hedonic games, we identify a number of classes of graphs which induce games that admit a non-empty core. These include graphs with degree at most two, forests, multi-partite complete graphs, bipartite graphs which admit perfect matchings, and graphs with girth at least five. For each of these classes, we also present polynomial-time algorithms to compute a core stable partition.
	\item We formulate the Bakers and Millers setting as a fractional hedonic game based on complete bipartite (or, more generally, complete $k$-partite) graphs. We show that such games always admit a non-empty strict core, and that the grand coalition is always stable. We characterize the partitions in the strict core, and give a polynomial time algorithm to compute a unique \emph{finest} partition in the strict core.
	\item We discuss how computing desirable outcomes in fractional hedonic games provides an interesting game-theoretic perspective to community detection~\citep[see, \eg][]{Newm04a,Fort10a} and network clustering.\footnote{\citet{CNM04a} discuss how social network analysis can be used to identify clusters of like-minded buyers and sellers in \texttt{amazon}'s purchasing network.}
\end{itemize}

\section{Related Work}

Fractional hedonic games are related to \emph{additively separable} hedonic games~\citep[see, \eg][]{Olse09a,SuDi10a,ABS11c}.
In both fractional hedonic games and additively separable hedonic games, each player ascribes a cardinal value to every other player. In additively separable hedonic games, utility in a coalition is derived by adding the values for the other players. By contrast, in fractional hedonic games, utility in a coalition is derived by adding the values for the other players and then dividing the sum by the total number of players in the coalition. 
Although conceptually, additively separable and fractional hedonic games are similar, their formal properties are quite different. As neither of the two models is obviously superior, this shows how slight modeling decisions may affect the formal analysis. 
For example, in unweighted and undirected graphs, the grand coalition is trivially core stable for additively separable hedonic games. On the other hand, this is not the case for fractional hedonic games.\footnote{Examples of this kind show that there are additively separable hedonic games which cannot be represented as a fractional hedonic game, and \emph{vice versa}.} A fractional hedonic game approach to social networks with only non-negative weights may help detect like-minded and densely connected communities.  In comparison, when the network only has non-negative weights for the edges, any reasonable solution for the corresponding additively separable hedonic game returns the grand coalition, which is not informative.

The difference between additively separable and fractional hedonic games is reminiscent of some issues in \emph{population ethics} (see, e.g., \citealp{ART17a}), which concerns the evaluation of states of the world with different numbers of individuals alive. Two prominent principles in population ethics are \emph{total utilitarianism} and \emph{average utilitarianism}. The former claims that a state of the world is better than another if it has a higher sum of individual utility, whereas the latter ranks states by the average utility enjoyed by the individuals. Many of the paradoxes of population ethics are analogous to properties of hedonic games. For example, average utilitarianism and fractional hedonic games both suffer from the `Mere Addition Paradox' \citep{Parf84a}, according to which a state of the world (resp., a coalition) can become less attractive if we add to it another positive-utility player (but whose utility is lower than the current average). Note, however, that this paradox cannot occur for simple and symmetric fractional hedonic games.

\citet{Olse12a} examined a variant of simple symmetric fractional hedonic games and investigated the computation and existence of Nash stable partitions. In the games he considered, however,
every maximal matching is core stable and every perfect matching is a best possible outcome, even if large cliques are present in the graph. By contrast, in our setting players have an incentive to form large cliques.

Fractional hedonic games are different from, but related to, another class of hedonic games called \emph{social distance games}, which were introduced by \citet{BrLa11a}. 
In social distance games, a player's utility from another player's presence in a coalition is inversely proportional to the distance between them in the subgraph induced by the coalition.
In many situations, one does not derive an additional benefit from friends of friends and may in fact prefer to minimize the fraction of people one does not agree with or have direct connections with. In such scenarios, fractional hedonic games are more suitable than social distance games.

Fractional hedonic games also exhibit some similarity with the segregation and status-seeking models considered by \citet{MiWi01a} and \citet{LaDi12a}. 
Group formation models based on types were first considered by \citet{Sche71a}. 

Independently from our work, \citet{FLN12a} have also considered the hedonic games framework as an approach to graph clustering. However, their research does not relate to core and strict core stability and they study different classes of hedonic games. 

Since their inception in the conference version of this paper \citep{ABH13a}, fractional hedonic games have already sparked some followup work.
\citet{AGG+15b} took a welfare maximization approach to fractional hedonic games and considered the complexity of finding partitions that maximize utilitarian or egalitarian social welfare.  
\citet{BFF+14a} analyze fractional hedonic games from the viewpoint of non-cooperative game theory. They show that Nash stable partitions may not exist in the presence of negative valuations. Furthermore, they give bounds on the price of anarchy and the price of stability.
\citet{BFF+15a} and \citet{KKP16a} further examine the price of (Nash) stability in simple symmetric fractional hedonic games, and \citet{EFF16a} consider the price of Pareto optimality. \citet{BBS14a} presented computational results for various stability concepts for fractional hedonic games. \citet{PeEl15a} identified structural features for various classes of hedonic games for which finding stable partitions is NP-hard. Their analysis implies several hardness results for fractional hedonic games.

\citet{LiWe17a} discuss simple symmetric fractional hedonic games (which they call \emph{popularity games}) as a model for the formation of \emph{socially cohesive groups}. They argue that in social networks, groups form based both on individual needs and desires, and on the group's resistance to disruption. Formally, individuals wish to maximize their \emph{popularity} in the group (measured by the fraction of the group that they are connected to in the network), while insisting that the group is core stable. \citet{LiWe17a} identify several classes of networks in which the grand coalition is core stable and give some necessary conditions in terms of structural cohesiveness measures. They also show that it is NP-hard to decide whether the grand coalition is core stable in a given simple symmetric fractional hedonic game, and present and evaluate some heuristics for this question.

\citet{WHN17a} use fractional hedonic games as a model of the formation of jurisdictions, noting that the arrangement of political boundaries involves a tradeoff between efficiencies of scale and of geographic heterogeneity. In examining the core of their weighted symmetric FHGs, they randomly sampled such games and found that all their samples admit a non-empty core, suggesting that the problem of non-existence of stable outcomes is not a problem in practice. They also introduce a heuristic algorithm for finding a core stable outcome, which proceeds by repeatedly searching for a blocking coalition (using integer programming) and myopically implementing the corresponding coalitional deviation. They then apply this algorithm to specific games modeled using historical data from Japan about political boundary changes, finding that their algorithm always found a core solution and typically terminated within a few hours, even for games containing approximately 1,000 players. They conclude that FHGs are ``an appropriate way of modeling mergers and splits of political jurisdictions'' and that they ``might also be used to model the formation of students into schools or classes, workers into unions, or public employees into different pension funds''.

\section{Preliminaries}\label{sec:prelim}

Let~$N$ be a set~$\set{1,\dots,n}$ of \emph{agents} or \emph{players}. A \emph{coalition} is a subset of the players. For every player~$i\in N$, we let~$\mathcal N_i$ denote the set $\set{S\subseteq N\colon i\in S}$ of coalitions~$i$ is a member of. 
Every player $i$ is equipped with a reflexive, complete, and transitive \emph{preference relation} $\pref_i$ over the set $\mathcal N_i$. We use~$\succ_i$ and~$\sim_i$ to refer to the strict and indifferent parts of $\succsim_i$, respectively. If~$\pref_i$ is also anti-symmetric we say that~$i$'s preferences are \emph{strict}. A coalition~$S\in\mathcal N_i$ is \emph{acceptable} for a player~$i$ if $i$ weakly prefers~$S$ to being alone, \ie $S\succsim_i{\set i}$, and \emph{unacceptable} otherwise.
A \emph{hedonic game} is a pair $(N,\pref)$, where ${\pref}=(\pref_1,\dots,\pref_{n})$ is a profile of preference relations~$\pref_i$, modeling the preferences of the players. 

A \emph{valuation function} of a player~$i$ is a function~$v_i\colon N\to\mathbb R$ assigning a real value to every player.
A hedonic game~$(N,\pref)$ is said to be a \emph{fractional hedonic game (FHG)} if, for every player~$i$ in~$N$, there is a valuation function~$v_i$  such that for all coalitions $S,T\in \mathcal N_i$,
\begin{align*}
	S\pref_i T \text{ if and only if } v_i(S)\ge v_i(T)\text,
\end{align*} 
where, for all $S\in\mathcal N_i$,
\begin{align*}
	v_i(S)	& =	\frac{\sum_{j\in S}v_i(j)}{|S|}\text.
\end{align*}
Hence, every FHG can be compactly represented by a tuple of valuation functions $v = (v_1,\dots,v_n)$.
It can be shown that every FHG can be induced by valuation functions with $v_i(i) = 0$ for all $i\in N$.\footnote{Let $v_i'(j) = v_i(j) - v_i(i)$ for all $i,j\in N$. Then, $v'_i(i) = 0$, for all $i\in N$ and $S\subseteq N$, $v_i(S) = \sum_{j\in S} \frac{v_i(j)}{|S|} = \sum_{j\in S} \frac{v'_i(j) + v_i(i)}{|S|} = \sum_{j\in S} \frac{v_i'(i)}{|S|} + v_i(i)$. Thus, for all $S,T\subseteq N$, $v_i(S) \ge v_i(T)$ if and only if $v'_i(S) \ge v_i'(T)$.} Thus, we assume $v_i(i) = 0$ throughout the paper.
We will frequently associate FHGs with weighted digraphs $G = (N,N\times N, v)$ where the weight of the edge $(i,j)$ is $v_i(j)$, \ie the valuation of player $i$ for player $j$. 

Two key restriction on the valuations in an FHG will be of particular interest to us.
\begin{itemize}
\item An FHG is \emph{symmetric} if $v_i(j)=v_j(i)$ for all $i,j\in N$.
\item An FHG is \emph{simple} if $v_i(j)\in\{0,1\}$ for all $i,j\in N$.
\end{itemize}

Simple FHGs have a natural appeal. Politicians may want to be in a party which maximizes the fraction of like-minded members and, for whatever reasons, people may want to be with as large a fraction of people of the same ethnic or social group. These situations can be fruitfully modeled and understood as a simple FHG by having  the players assign value~$1$ to like-minded or ethnically similar people, and~$0$ to others.

A simple FHG $(N,\succsim_i)$ can be represented by a digraph~$(V,A)$ in which $V=N$ and $(i,j)\in A$ if and only if $v_i(j)=1$. Similarly, if $(N,\succsim_i)$ is both symmetric and simple, it can be represented by an (undirected) graph $(V,E)$ such that $V=N$ and $\set{i,j}\in E$ if and only if $v_i(j)=v_j(i)=1$. With this representation, we will often think of graphs and simple symmetric FHGs as the same thing.

The outcomes of hedonic games are \emph{partitions} of the players, also known as \emph{coalition structures}.  Given a partition~$\pi=\set{S_1,\dots,S_m}$ of the players,~$\pi(i)$ denotes the coalition in~$\pi$ of which player~$i$ is a member. We also write~$v_i(\pi)$ for $v_i(\pi(i))$, which is the utility that $i$ receives in $\pi$, reflecting the hedonic nature of the games we consider. By the same token we obtain preferences over partitions from preferences over coalitions.
We refer to $\{N\}$ as the \emph{grand coalition}.

Hedonic games are analyzed using \emph{solution concepts}, which formalize desirable or optimal ways in which the players can be partitioned (as based on the players' preferences over the coalitions). If a partition satisfies a given solution concept, it is considered to be \emph{stable} in the sense of the solution concept. 
	A basic requirement for partitions to be acceptable for all players is individual rationality. A partition $\pi$ is \emph{individually rational} if each player weakly prefers his coalition in $\pi$ over being alone, \ie for each $i\in N$, $\pi(i)\pref_i \{i\}$. Intuitively, if a partition is not individually rational, it cannot be stable, since one player has an incentive to leave his current coalition and be on his own instead.
	In this paper, we will focus on two of the most prominent solution concepts, the \emph{core} and the \emph{strict core}, taken from cooperative game theory. We say that a coalition $S \subseteq N$ \emph{(strongly) blocks} a partition~$\pi$, if each
	player $i \in S$ strictly prefers~$S$ to his current coalition~$\pi(i)$ in
	the partition~$\pi$, \ie if $S\succ_i\pi(i)$ for all $i\in S$. A partition that does not admit a blocking coalition is said to be in the \emph{core}. In a similar vein, we say that a coalition $S \subseteq N$ \emph{weakly blocks} a partition~$\pi$,
	if each player $i \in S$ weakly prefers~$S$ to~$\pi(i)$ and there exists at least one player $j \in S$ who strictly prefers~$S$ to his
	current coalition~$\pi(j)$, \ie $S\succsim_i\pi(i)$ for all $i\in N$ and $S\succ_j\pi(j)$ for some~$j\in S$. A partition does not admit a weakly blocking coalition is in the \emph{strict core}.
	Cleary, the strict core is a subset of the core. Moreover, the core is a subset of the set of individually rational coalitions, since every coalition that is not individually rational is blocked by a singleton coalition.

\begin{example}
		Consider the simple and symmetric FHG based on the graph depicted in Figure~\ref{fig:ex:hedonic_4}. In the grand coalition, the utility of each player is $\frac{1}{2}$. There is only one core stable partition: $\set{\set{1,2,3},\set{4,5,6}}$, which yields utility $\frac{2}{3}$ for each player.
Observe that, when interpreted as an additively separable hedonic game, this is not a stable partition, as the grand coalition would yield a higher utility---namely,~$3$ instead of~$2$---to all and thus be a blocking coalition.

	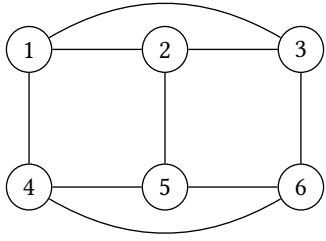
\begin{figure}[t]
	\centering
	\scalebox{1}{
	\begin{tikzpicture}[<->,auto, line width = 0.5pt]

		\def \radius {12ex}

		\draw 
			node[circle,fill=white,draw](1){$1$}
			++(0:\radius)
			node[circle,fill=white,draw](2){$2$}
			++(0:\radius)
			node[circle,fill=white,draw](3){$3$}
			++(-90:\radius)
			node[circle,fill=white,draw](6){$6$}
			++(180:\radius)
			node[circle,fill=white,draw](5){$5$}
			++(180:\radius)
			node[circle,fill=white,draw](4){$4$}			
		;

\path[draw,line width = 0.5pt,-] (1) -- (2);
\path[draw,line width = 0.5pt,-] (1) -- (4);
\path[draw,line width = 0.5pt,-] (2) -- (3);
\path[draw,line width = 0.5pt,-] (2) -- (5);
\path[draw,line width = 0.5pt,-] (4) -- (5);
\path[draw,line width = 0.5pt,-] (6) -- (5);
\path[draw,line width = 0.5pt,-] (6) -- (3);

\path[draw,line width = 0.5pt,-]	(3) edge [bend right] (1);
\path[draw,line width = 0.5pt,-]	(4) edge [bend right] (6);
	\end{tikzpicture}
	}

	\caption{Example of a simple and symmetric FHG. The only core stable partition is $\set{\set{1,2,3},\set{4,5,6}}$. In contrast, if the graph represents an additively separable hedonic game, then the partition consisting of the grand coalition is the only core stable partition.}\label{fig:ex:hedonic_4}
	\end{figure}
\end{example}

Some standard graph-theoretic terminology will be useful. The complete undirected graph on $n$ vertices is denoted by~$K_n$; an undirected cycle on $n$ vertices is denoted by~$C_n$. 
A graph~$(V,E)$ is said to be \emph{$k$-partite} if~$V$ can be partitioned into~$k$ independent sets $V_1,\dots,V_k$, \ie $v,w\in V_i$ implies $\set{v,w}\notin E$. A $k$-partite graph is \emph{complete} if for all $v\in V_i$ and $w\in V_j$ we have $\set{v,w}\in E$ if and only if $i\neq j$.

\section{Negative Results}
\label{sec:negative}

For any game-theoretic solution concept, two natural questions are whether a solution is always guaranteed to exist, and whether a solution can be found efficiently. For the core of FHGs, the answer to both of these questions turns out to be negative if we do not restrict the structure of the underlying graph.

For unrestricted FHGs (that is, if we allow any weighted digraph), it is easy to construct examples whose core is empty (see \figref{fig:core}). Even if we require the game to be symmetric, it is not difficult to find examples with empty core (see \figref{fig:symcore}). Of course, the examples given are specifically constructed so as to not admit a core stable outcome, and it is plausible that ``most'' FHGs do admit one. Indeed, \citet{WHN17a} randomly sampled 10 million symmetric FHGs similar to the one shown in \figref{fig:symcore}, and all of them had a non-empty core.

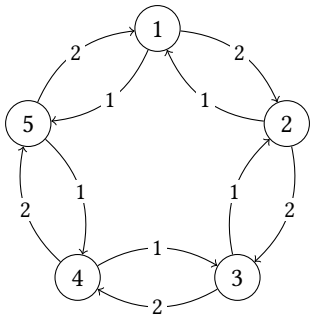
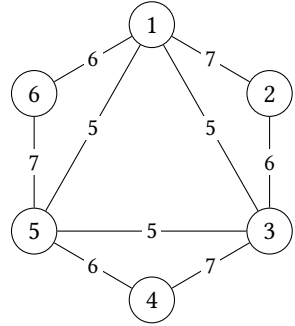
\begin{figure}[t]
	\centering
	\begin{subfigure}[b]{0.42\textwidth}
		\centering
		\begin{tikzpicture}
		
		\def \radius {12ex}
		
		\draw
		(0,0)
		+(90:\radius)
		node[circle,fill=white,draw](1){$1$}
		+(90+72:\radius)
		node[circle,fill=white,draw](2){$5$}
		+(90+2*72:\radius)
		node[circle,fill=white,draw](3){$4$}
		+(90+3*72:\radius)
		node[circle,fill=white,draw](4){$3$}
		++(90+4*72:\radius)
		node[circle,fill=white,draw](5){$2$}
		;
		
		\path[every node/.style={font=\sffamily\small},->]
		(1) edge [bend left] node[fill=white,inner sep=1.5pt]  {$1$} (2)
		(2) edge [bend left] node[fill=white,inner sep=1.5pt]  {$1$} (3)
		(3) edge [bend left] node[fill=white,inner sep=1.5pt]  {$1$} (4)
		(4) edge [bend left] node[fill=white,inner sep=1.5pt]  {$1$} (5)
		(5) edge [bend left] node[fill=white,inner sep=1.5pt]  {$1$} (1)
		;
		
		\path[every node/.style={font=\sffamily\small},->]
		(1) edge [bend left] node[fill=white,inner sep=1.5pt]  {$2$} (5)
		(2) edge [bend left] node[fill=white,inner sep=1.5pt]  {$2$} (1)
		(3) edge [bend left] node[fill=white,inner sep=1.5pt]  {$2$} (2)
		(4) edge [bend left] node[fill=white,inner sep=1.5pt]  {$2$} (3)
		(5) edge [bend left] node[fill=white,inner sep=1.5pt]  {$2$} (4)
		;
		
		\end{tikzpicture}
		\caption{An FHG given by a weighted digraph whose core is empty. All missing edges have weight $-10$.}
		\label{fig:core}
	\end{subfigure}\qquad%
	\begin{subfigure}[b]{0.425\textwidth}
		\centering
		\begin{tikzpicture}
		
		\def \n {6}
		\def \radius {12ex}

		\node[circle,draw] (1) at ({360/\n * (1.5)}:\radius) {$1$};
		\node[circle,draw] (2) at ({360/\n * (0.5)}:\radius) {$2$};
		\node[circle,draw] (3) at ({360/\n * (5.5)}:\radius) {$3$};
		\node[circle,draw] (4) at ({360/\n * (4.5)}:\radius) {$4$};
		\node[circle,draw] (5) at ({360/\n * (3.5)}:\radius) {$5$};
		\node[circle,draw] (6) at ({360/\n * (2.5)}:\radius) {$6$};

		\path[every node/.style={font=\sffamily\small}]
		(1) edge node[fill=white,inner sep=1.5pt] {$7$} (2)
		edge node[fill=white,inner sep=1.5pt] {$5$} (3)
		(2) edge node[fill=white,inner sep=2pt] {$6$} (3)
		(3) edge node[fill=white,inner sep=1.5pt] {$7$} (4)
		edge node[fill=white,inner sep=1.5pt] {$5$} (5)
		(4) edge node[fill=white,inner sep=1.5pt] {$6$} (5)
		(5) edge node[fill=white,inner sep=2pt] {$7$} (6)
		edge node[fill=white,inner sep=1.5pt] {$5$} (1)
		(6) edge node[fill=white,inner sep=1.5pt] {$6$} (1)
		;
		\end{tikzpicture}
		
		\caption{A symmetric FHG given by weighted graph whose core is empty. All missing edges have weight $-24$.}
		\label{fig:symcore}
	\end{subfigure}
	\caption{Examples of (symmetric) FHGs with empty core.}
	\label{fig:emptycore}
\end{figure}

If we consider the \emph{strict} core, it is also easy to construct an example of a simple symmetric FHG whose strict core is empty: consider the FHG represented by a cycle of size five ($C_5$). For $C_5$, any coalition of size three or more admits a blocking coalition of size two, and partitions consisting of one singleton and two coalitions of size two also admit a weakly blocking coalition.

It was open for some time whether there is a simple symmetric FHG whose core is empty. Here, we present such an example, consisting of a total of 40 players, see Figure~\ref{fig:simsymcore}. It is unclear whether smaller examples exist. Note that this result subsumes all of the non-existence results mentioned above.

\begin{figure}[t]
	\centering
	\begin{tikzpicture}[-, line width = 0.5pt]
	
	\tikzset{
		position/.style args={#1:#2 from #3}{
			at=(#3.#1), anchor=#1+180, shift=(#1:#2)
		}
	}
	
	\def \n {10}
	\def \k {9}
	\def \radius {18ex}
	\def \nodedist {.9ex}
	\def \margin {8} 
	
	\draw (0,0) circle (\radius);

	\node[circle,draw,minimum size = 6.5*\nodedist,fill = white] (a_1) at ({360/\n * (2.5)}:\radius) {};
	\node[circle,draw,minimum size = 6.5*\nodedist,fill = white] (a_2) at ({360/\n * (0.5)}:\radius) {};
	\node[circle,draw,minimum size = 6.5*\nodedist,fill = white] (a_3) at ({360/\n * (8.5)}:\radius) {};
	\node[circle,draw,minimum size = 6.5*\nodedist,fill = white] (a_4) at ({360/\n * (6.5)}:\radius) {};
	\node[circle,draw,minimum size = 6.5*\nodedist,fill = white] (a_5) at ({360/\n * (4.5)}:\radius) {};
	
	\node[circle,draw,minimum size = 6.5*\nodedist,fill = white] (b_1) at ({360/\n * (1.5)}:\radius) {};
	\node[circle,draw,minimum size = 6.5*\nodedist,fill = white] (b_2) at ({360/\n * (9.5)}:\radius) {};
	\node[circle,draw,minimum size = 6.5*\nodedist,fill = white] (b_3) at ({360/\n * (7.5)}:\radius) {};
	\node[circle,draw,minimum size = 6.5*\nodedist,fill = white] (b_4) at ({360/\n * (5.5)}:\radius) {};
	\node[circle,draw,minimum size = 6.5*\nodedist,fill = white] (b_5) at ({360/\n * (3.5)}:\radius) {};
	
	\node[circle,draw,minimum size = 6.5*\nodedist] (c_1) at ($(a_1)!0.5!(b_4)$) {};
	\node[circle,draw,minimum size = 6.5*\nodedist] (c_2) at ($(a_2)!0.5!(b_5)$) {};
	\node[circle,draw,minimum size = 6.5*\nodedist] (c_3) at ($(a_3)!0.5!(b_1)$) {};
	\node[circle,draw,minimum size = 6.5*\nodedist] (c_4) at ($(a_4)!0.5!(b_2)$) {};
	\node[circle,draw,minimum size = 6.5*\nodedist] (c_5) at ($(a_5)!0.5!(b_3)$) {};
	
	\foreach \x in {1,2,3,4,5}
	{
		
		\foreach \y in {1,2,3}
		{
			\node[circle,draw,position=42+120*\y-72*\x:{-2.7*\nodedist} from a_\x] (a_\x\y) {};
			\node[circle,draw,position=54+42+120*\y-72*\x:{-2.7*\nodedist} from c_\x] (c_\x\y) {};
			
		}
		
		\node[circle,draw,position=0-36+72-72*\x:{-2.7*\nodedist} from b_\x] (b_\x1) {};
		\node[circle,draw,position=180-36+72-72*\x:{-2.7*\nodedist} from b_\x] (b_\x2) {};

		\draw (a_\x1) -- (a_\x2);
		\draw (a_\x1) -- (a_\x3);	
		\draw (a_\x3) -- (a_\x2);
		
		\draw (c_\x1) -- (c_\x2);
		\draw (c_\x1) -- (c_\x3);	
		\draw (c_\x3) -- (c_\x2);
		
		\draw (b_\x1) -- (b_\x2);
		
	}
	
	\path[every node/.style={font=\sffamily\small}]
	
	(a_1) edge [near end] node {} (c_1) 
	(a_2) edge [near end] node {} (c_2) 
	(a_3) edge [near end] node {} (c_3) 
	(a_4) edge [near end] node {} (c_4) 
	(a_5) edge [near end] node {} (c_5) 
	
	(b_1) edge [near end] node {} (c_3) 
	(b_2) edge [near end] node {} (c_4) 
	(b_3) edge [near end] node {} (c_5) 
	(b_4) edge [near end] node {} (c_1) 
	(b_5) edge [near end] node {} (c_2) 
	;
	\end{tikzpicture}
	\caption{A simple symmetric FHG with 40 players whose core is empty (see Theorem~\ref{thm:simsymcore}). An edge between two encircled cliques indicates that every vertex of one clique is connected to every vertex of the other.}	
	\label{fig:simsymcore}
\end{figure}
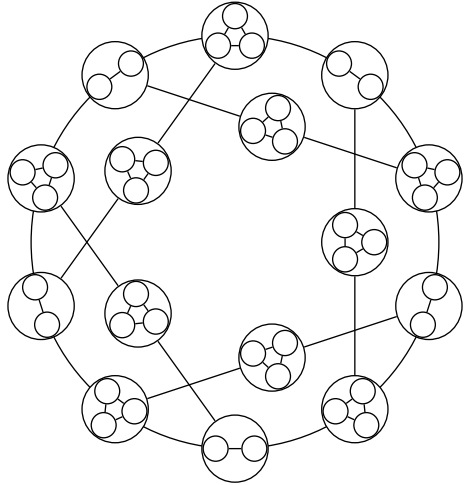

\begin{theorem}\label{thm:simsymcore}
	In simple and symmetric FHGs, the core can be empty.
\end{theorem}

The proofs of this and other results can be found in the appendix. \\

Now that we have seen that the core of an FHG can be empty, we can move on to some computational questions. The natural problem to consider is to \emph{find} a core stable partition. Since such an partition does not always exist, we can conveniently consider the \emph{decision} problem of whether the core of a given FHG is non-empty. It turns out that, without imposing further restrictions, answering this question is computationally difficult, and in particular NP-hard. The examples with empty core that we have seen make for convenient gadgets in hardness reductions. 

Notice that the problem of checking whether an FHG has non-empty core is not obviously contained in the class NP: the natural certificate would be a core stable partition, but it is not at all clear how to check whether a given partition is in the core; naively, this would require checking all $2^n$ possible blocking coalitions. In fact, this problem of verifying whether a partition is in the core is already coNP-complete (we sketch a proof in the appendix, an alternative proof appears in \citealp{LiWe17a}). The natural complexity class for the non-emptiness problem is $\Sigma_2^p = \text{NP}^\text{NP}$ which captures the alternation of quantifiers: ``does there \emph{exist} a partition $\pi$ such that \emph{for all} coalitions $S$, the coalition does not block?''. Indeed, we can show that the non-emptiness problem is complete for this class, even for simple and symmetric FHGs.

\begin{theorem}\label{thm:simsymcoreNP}
	Checking whether a simple symmetric FHG has an empty core is $\Sigma_2^p$-complete.
\end{theorem}

The proof of this statement is a rather involved reduction from the MAX-MIN CLIQUE problem \citep{KoLi95a}, and uses a notion of `subsidies' to players who are put in singleton coalitions. FHGs are not the only class of hedonic games for which checking non-emptiness of the core is $\Sigma_2^p$-complete: additively separable and Boolean hedonic games are other examples \citep{Woeg13a,Pete15a}.

As \citet{Woeg13a} argues, the fact that finding a core stable outcome is $\Sigma_2^p$-hard means that this solution concept is computationally much harder to handle than solution concepts like Nash stability, where the analogous decision problem is contained in NP. Indeed, recent advances in SAT and ILP solvers mean that many NP-complete problems of moderate size are now easily solvable in practice; this is not the case for $\Sigma_2^p$-complete problems.

However, \citet{WHN17a} present a heuristic algorithm that attempts to find a core stable partition by repeatedly searching for a blocking coalition using an ILP solver and implementing this deviation. They find that this approach is reasonably efficient for real-world examples with up to 1,000 players.
		
\section{Positive Results}
\label{sec:positive}

In this section, we present a number of subclasses of simple and symmetric FHGs for which the core is non-empty. Since these games can be represented by unweighted and undirected graphs, we will focus on different graph classes. 
In particular we show existence results for the following classes of graphs: 
graphs with degree at most two, forests, multi-partite complete graphs, bipartite graphs which admit perfect matchings, regular bipartite graphs, and graphs with girth at least five.
All of our proofs are constructive in the sense that we show that the core is non-empty by outlining a way to construct a partition in the core; in each case this construction can be performed in polynomial time. 
	
\subsection{Graphs with bounded degree}

If a graph is extremely sparse, then intuitively it does not admit interesting blocking coalitions. Indeed, we have the following result.

\begin{theorem}
	For simple and symmetric FHGs represented by graphs of degree at most 2, the core is non-empty. 	
\end{theorem}

The proof employs a simple greedy algorithm partitioning the players into coalitions of size at most 3. Such a strategy is successful in this case since the connected components of graphs of degree at most 2 are paths and cycles, and in this situation, players are relatively happy in a small coalition together with an immediate neighbor.

Theorem~\ref{thm:simsymcore} shows that the positive result for the degree bound of 2 cannot be extended to a bound of 11 (which is the maximum degree of the example given there). It might be interesting to close this gap; but the case of degree 3 already seems difficult.

\subsection{Forests} 

The example of an FHG with empty core that we gave above depended crucially on an underlying cyclic structure of the game. If we do not allow such cycles, the problem disappears:

\begin{theorem}\label{th:trees}
	For simple and symmetric FHGs represented by undirected forests, the core is non-empty.
\end{theorem}

The proof employs a simple dynamic programming algorithm, which essentially exploits that the preferences of a vertex are somewhat opposed to the preferences of its grandparent, so that blocking coalitions would need to be very `local'. The algorithm matches generations into pairs and produces a partition in which vertices are locally satisfied, making it stable.

The conclusion of Theorem~\ref{th:trees} could be reached in an alternative way. Notice that the coalitions in any partition in the core of a simple and symmetric FHG need to be \emph{connected} in the underlying graph: otherwise, a proper connected component would block. Thus, an FHG given by a graph $G$ can be viewed as a \emph{hedonic game with graph structure} with communication structure given by $G$ in the sense of \citet{IgEl16a}. They showed that, by a result due to \citet{Dema04a}, the core of such games is non-empty if $G$ is a forest. However, this method does not yield a polynomial-time algorithm to produce an element of the core.

\subsection{Bakers and Millers: complete $k$-partite graphs}

In the introduction we referred to the \emph{Bakers and Millers} setting, in which the players are of two different types and each of them prefers the fraction of players of the other type to be as high as possible. The setting could arise if individuals of the same type are competitors engaging in trade with individuals of the other type. 

This idea can easily be extended to multiple types.
Let  $\types=\{\type_1,\ldots,\type_k\}$ be a set of types that partitions the set~$N$ of players, where $k=|\types|$. Let $\type(i)$ denote the type of player~$i$. A hedonic game $(N,\succsim)$ is called a \emph{Bakers and Millers game} if the preferences of each player~$i$ are such that for all coalitions $S,T\in\mathcal N_i$, 
\[
	\text{
	$S\mathrel{\pref_i}T$\quad if and only if\quad $\frac{|S\cap \type(i)|}{|S|}\le \frac{|T\cap \type(i)|}{|T|}$.}
\]
Thus, a player prefers coalitions in which a larger fraction of players are of a different type.
With this formalization, we see that a Bakers and Millers game with~$k$ types is a simple and symmetric FHG represented by a \emph{complete $k$-partite graph} with the maximal independent sets representing the types, \ie a graph~$(V,E)$ with $V=N$ and 
\[
	E	= \set{\set{i,j}\midd \type(i)\neq\type(j)}\text.
\]

In this setting, the grand coalition is always in the strict core: Since the types partition the player set, observe that for every coalition~$S$ we have
\[
	\frac{|S\cap \type_1|}{|S|}+\dots+\frac{|S\cap \type_t|}{|S|}=1\text.
\]
Now assume for a contradiction that the grand coalition~$N$ is not in the strict core. Then there is a (weakly blocking) coalition~$S$ such that $\frac{|S\cap \type(i)|}{|S|}< \frac{|N\cap \type(i)|}{|N|}$ for some $i\in S$ and $\frac{|S\cap \type(j)|}{|S|}\le\frac{|N\cap \type(j)|}{|N|}$ for all $j\in S$. But then 
\[
	\frac{|S\cap \type_1|}{|S|}+\dots+\frac{|S\cap \type_t|}{|S|}<	\frac{|N\cap \type_1|}{|N|}+\dots+\frac{|N\cap \type_t|}{|N|}\text{,}
\]
\ie $1<1$, a contradiction. Generalizing this idea we obtain the following theorem. 

\begin{theorem}\label{thm:corestable1}
	Let $(N,\succsim)$ be a Bakers and Millers game with type space
	 $\types=\set{\type_1,\dots,\type_t}$ and~$\cstr=\set{S_1,\dots,S_m}$  a partition. Then, $\cstr$ is in the strict core if and only if for all types~$\type\in\types$ and all coalitions~$S,S'\in\pi$,  \[\frac{|S\cap\type|}{|S|}=\frac{|S'\cap\type|}{|S'|}\text.\]
\end{theorem}

Let~$d$ denote the greatest common divisor of $|\theta_1|,\dots,|\theta_{t}|$, which we know can be computed in time linear in~$t$ \citep[cf.][]{Brad70a}. Theorem~\ref{thm:corestable1} can now be rephrased as follows: a partition~$\pi$ for a Bakers and Millers game is in the strict core if and only if for all coalitions~$S$ in~$\pi$ there is a positive integer~$k_S$ such that for all types~$\theta_i$ we have $|S\cap\theta_i|=k_S|\theta_i|/{d}$. Thus, for the grand coalition~$N$ we have $k_N=d$. There is also a partition~$\pi$ such that $k_S=1$ for all coalitions~$S$ in~$\pi$ which is in the strict core; no finer partition is in the strict core.
We say that two partitions~$\pi$ and~$\pi'$ are \emph{identical up to renaming players of the same type} if there is a bijection $f\colon N\to N$ 
such that for all players~$i$ we have $\theta(i)=\theta(f(i))$ and
$\pi'=\set{\set{f(i)\midd i\in S}\midd S\in\pi}$.
Hence, we have the following corollary.

\begin{corollary}
	For every Bakers and Millers game, there is a \emph{unique} finest partition in the strict core (up to renaming players of the same type), which, moreover, can be computed in linear time.
\end{corollary}

As the strict core is a subset of the core, the ``if''-direction of \thmref{thm:corestable1} also holds for the core: every partition~$\pi$ such that  $\frac{|S\cap\type|}{|S|}=\frac{|S'\cap\type|}{|S'|}$ for all types~$\type\in\types$ and all coalitions~$S,S'\in\pi$ is in the core. The inverse of this statement, however, does not generally hold:
Consider three players~$1$,~$2$, and~$3$, with~$1$ belonging to type~$\type_1$, while~$2$ and~$3$ belong to type~$\type_2$. Then, the coalition structure $\{\{1,2\}, \{3\}\}$ is in the core but not in the strict core: the coalition $\set{1,3}$ would be weakly blocking.

\subsection{Graphs with large girth}

The \emph{girth} of a graph is the length of the shortest cycle in the graph. For example, bipartite graphs have a girth of at least four. 
Graphs with a girth of at least five do not admit triangles or cycles of length four. FHGs described by such graphs always admit a core partition. The key idea behind this result is to pack the vertices of the graph representing a fractional game into stars while maximizing a particular objective function (namely, maximizing the \emph{leximin objective}). The resulting partition is in the core.

\begin{theorem}
	For simple and symmetric FHGs represented by graphs with girth at least five, the core is non-empty, and there always exists a partition into stars that is in the core.
\end{theorem}

The argument establishing this, while constructive, does not directly yield a polynomial-time algorithm finding an element of the core, since it is not clear whether a star packing optimizing our leximin objective function can be found polynomial time. In the appendix, we show that certain local maxima are also core stable, and that they can be found in polynomial time through local search.

It is worth observing that, even if a core partition consisting of stars exists, there may still be other stable partitions. Consider, for example, a game given by a star with 3 leaves on the vertex set $\{ c , \ell_1,\ell_2,\ell_3\}$. Then the partition $\pi = \{ \{c, \ell_1, \ell_2\}, \{\ell_3\} \}$ is in the core, but it is not a star packing.

\begin{figure}[t]
	\centering
		\scalebox{.65}{
			\LARGE
			\begin{tikzpicture}
			\tikzstyle{pfeil}=[--,draw]
			\tikzstyle{pfeil}=[->,>=angle 60, shorten >=1pt,draw]
			\tikzstyle{onlytext}=[]
			\tikzstyle{node}=[circle,fill=gray!45]

			\draw 
			(-10,0)
			node[circle,fill=white,draw,inner sep=.4ex](a){$a$}
			++(270:3) 
			+(180:1.5) node[circle,fill=white,draw,inner sep=.4ex](0){$0$}
			+(  0:1.5) node[circle,fill=white,draw,inner sep=.4ex](1){$1$}
			;                      
			
			\draw 
			(-5,0)
			node[circle,fill=white,draw,inner sep=.4ex](b){$b$}
			++(270:3) 
			+(180:1.5) node[circle,fill=white,draw,inner sep=.4ex](2){$2$}
			+(  0:1.5) node[circle,fill=white,draw,inner sep=.4ex](3){$3$}
			;                      
			\draw
			(0,0)
			node[circle,fill=white,draw,inner sep=.4ex](c){$c$}
			++(270:3)                     
			node[circle,fill=white,draw,inner sep=.4ex](5){$5$}
			+(180:1.7)                                      
			node[circle,fill=white,draw,inner sep=.4ex](4){$4$}
			+(0:1.7)                                        
			node[circle,fill=white,draw,inner sep=.4ex](6){$6$}
			;
			\draw 
			(5,0)
			node[circle,fill=white,draw,inner sep=.4ex](d){$d$}
			++(270:3) 
			node[circle,fill=white,draw,inner sep=.4ex](8){$8$}
			+(180:1.7)                                       
			node[circle,fill=white,draw,inner sep=.4ex](7){$7$}
			+(0:1.7)                                        
			node[circle,fill=white,draw,inner sep=.4ex](9){$9$}
			;
			
			\path[draw,very thick,-] (a) -- (0);
			\path[draw,very thick,-] (a) -- (1);      
			
			\path[draw,dotted,-]	(a)--(2);
			\path[draw,dotted,-]	(a)--(3);
			\path[draw,dotted,-]	(a)--(4);
			\path[draw,dotted,-]	(a)--(5);
			\path[draw,dotted,-]	(a)--(6);
			\path[draw,dotted,-]	(a)--(7);
			\path[draw,dotted,-]	(a)--(8);
			\path[draw,dotted,-]	(a)--(9);

			\path[draw,very thick,-] (b) -- (2);
			\path[draw,very thick,-] (b) -- (3);      
			
			\path[draw,dotted,-]	(b)--(0);
			\path[draw,dotted,-]	(b)--(1);
			\path[draw,dotted,-]	(b)--(4);
			\path[draw,dotted,-]	(b)--(5);
			\path[draw,dotted,-]	(b)--(6);
			\path[draw,dotted,-]	(b)--(7);
			\path[draw,dotted,-]	(b)--(8);
			\path[draw,dotted,-]	(b)--(9);

			\path[draw,very thick,-] (c) -- (4);
			\path[draw,very thick,-] (c) -- (5);
			\path[draw,very thick,-] (c) -- (6);
			
			\path[draw,dotted,-]	(c)--(2);
			\path[draw,dotted,-]	(c)--(3);
			\path[draw,dotted,-]	(c)--(0);
			\path[draw,dotted,-]	(c)--(1);
			\path[draw,dotted,-]	(c)--(7);
			\path[draw,dotted,-]	(c)--(8);
			\path[draw,dotted,-]	(c)--(9);

			\path[draw,very thick,-] (d) -- (9);
			\path[draw,very thick,-] (d) -- (7);
			\path[draw,very thick,-] (d) -- (8);
			
			\path[draw,dotted,-]	(d)--(2);
			\path[draw,dotted,-]	(d)--(3);
			\path[draw,dotted,-]	(d)--(0);
			\path[draw,dotted,-]	(d)--(1);
			\path[draw,dotted,-]	(d)--(4);
			\path[draw,dotted,-]	(d)--(5);
			\path[draw,dotted,-]	(d)--(6);

			\end{tikzpicture}
		}
		
		\caption{The complete bipartite graph $K_{4,10}$ in which no star packing yields a stable partition. For instance, the partition indicated by the solid edges is not stable as $\set{a,b,4,5,6,7,8}$ would deviate.}
		\label{fig:no_stable_star_packing}
	\end{figure}
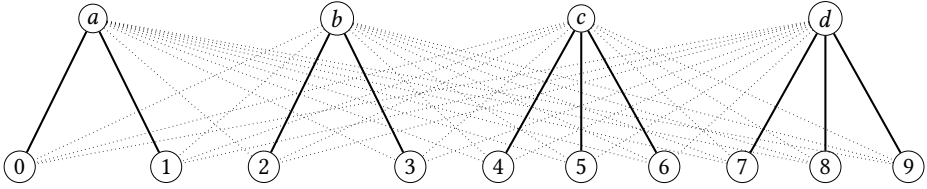

\subsection{Bipartite graphs}

For FHGs on bipartite graphs (whose girth is always at least 4), it is not always the case that there are star-packings that also yield partitions in the core: it can be checked that no partition into stars of the FHG given by the complete bipartite graph $K_{4,10}$ with 14 vertices is in its core; see \figref{fig:no_stable_star_packing}. On the other hand, since $K_{4,10}$ is a Bakers and Millers game, by \thmref{thm:corestable1}, the grand coalition is in the core, which thus is non-empty. This example shows that any method of finding partitions in the core in games on bipartite graphs based on finding star-packings is bound to fail.

We have not been able to prove whether the core is non-empty for all bipartite graphs, and this remains an interesting open problem. For certain subclasses of bipartite graphs, positive results can still be obtained. For example, we can observe that perfect matchings, if they exist, are in the core.
\begin{lemma}\label{lem:perfect_matchings}
	For every FHG that is represented by an undirected bipartite graph admitting a perfect matching the core is non-empty.
\end{lemma}

\begin{proof}
	Let $\set{N',N''}$ be the respective bipartition of~$N$. For every coalition $S\subseteq N$, either $\frac{|N'\cap S|}{|S|}\le {1\over 2}$ or $\frac{|N''\cap S|}{|S|}\le {1\over 2}$. Hence, every coalition~$S$ contains at least one player~$i$ with $v_i(S)\le {1\over 2}$. In a perfect matching, considered as a partition, every player has value $1\over 2$. Hence, every perfect matching is in the core and the claim follows.
\end{proof}
We can then obtain the following result as a corollary of Hall's Theorem.

\begin{corollary}
	For all bipartite $k$-regular graphs the core of the corresponding FHG is non-empty.
\end{corollary}

It would be desirable to find additional examples of classes of bipartite graphs for which a non-empty core is guaranteed.

\section*{Acknowledgments}

This material is based upon work supported by
the Australian Government's
Department of Broadband, Communications and the Digital
Economy, the Australian Research Council, the Asian
Office of Aerospace Research and Development through
grant AOARD-124056, and the Deutsche Forschungsgemeinschaft under grants {BR~2312/7-1} and {BR~2312/10-1}. Paul Harrenstein is supported by the~ERC under Advanced Grant~291528 (``RACE''). Dominik Peters is supported by~EPSRC.

Preliminary results of this paper were presented at the 
Blankensee Colloquium on Neighborhood Technologies in Berlin (August, 2012),
the 13th International Conference on Autonomous Agents and Multi-Agent Systems in Paris (May, 2014), and the 14th International Conference on Autonomous Agents and Multi-Agent Systems in Istanbul (May, 2015).

\bibliographystyle{ACM-Reference-Format}

\newpage

\setcounter{theorem}{0}

\appendix

	\section{Proofs}

	\subsection{Empty core}
	
	To prove \thmref{thm:simsymcore}, we give a simple and symmetric FHG that does not admit a core stable partition.
	Since this game is fairly large (40 players), we first illustrate the construction by giving a simpler example (15 players) from a slightly larger class of games.
	To this end, we say that an FHG is \emph{social} if $v_i(j)\ge 0$ for all $i,j\in N$.
	Clearly, every simple FHG is also social.
	The FHG depicted in \figref{fig:socsymcore} is social and symmetric, but has an empty core.
	We omit the proof, since \thmref{thm:simsymcore} proves a stronger statement.
	
	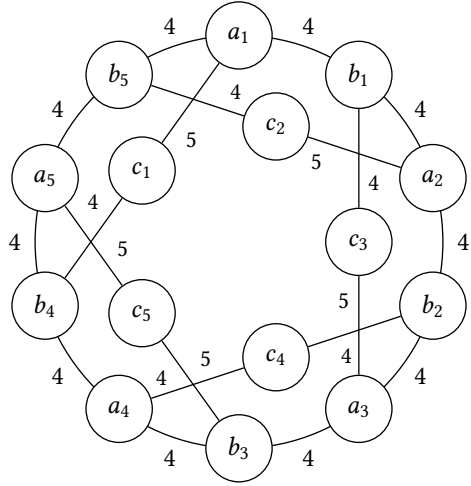
\begin{figure}
			
	\centering
	
		\begin{tikzpicture}[-, line width = 0.5pt, auto]

		\tikzset{
		    position/.style args={#1:#2 from #3}{
		        at=(#3.#1), anchor=#1+180, shift=(#1:#2)
		    }
		}


		\def \n {10}
		\def \k {9}
		\def \radius {18ex}
		\def \nodedist {.9ex}
		\def \margin {8} 
		\def \cliqueopacity {1} 
		\def \labelposition {1.25} 
		\def \labelpositionC {.85}

		\draw (0,0) circle (\radius);

		\node[circle,draw,minimum size = 6.5*\nodedist,fill = white] (a_1) at ({360/\n * (2.5)}:\radius) {$a_1$};
		\node[circle,draw,minimum size = 6.5*\nodedist,fill = white] (a_2) at ({360/\n * (0.5)}:\radius) {$a_2$};
		\node[circle,draw,minimum size = 6.5*\nodedist,fill = white] (a_3) at ({360/\n * (8.5)}:\radius) {$a_3$};
		\node[circle,draw,minimum size = 6.5*\nodedist,fill = white] (a_4) at ({360/\n * (6.5)}:\radius) {$a_4$};
		\node[circle,draw,minimum size = 6.5*\nodedist,fill = white] (a_5) at ({360/\n * (4.5)}:\radius) {$a_5$};

		\node[circle,draw,minimum size = 6.5*\nodedist,fill = white] (b_1) at ({360/\n * (1.5)}:\radius) {$b_1$};
		\node[circle,draw,minimum size = 6.5*\nodedist,fill = white] (b_2) at ({360/\n * (9.5)}:\radius) {$b_2$};
		\node[circle,draw,minimum size = 6.5*\nodedist,fill = white] (b_3) at ({360/\n * (7.5)}:\radius) {$b_3$};
		\node[circle,draw,minimum size = 6.5*\nodedist,fill = white] (b_4) at ({360/\n * (5.5)}:\radius) {$b_4$};
		\node[circle,draw,minimum size = 6.5*\nodedist,fill = white] (b_5) at ({360/\n * (3.5)}:\radius) {$b_5$};

		\node[circle,draw,minimum size = 6.5*\nodedist] (c_1) at ($(a_1)!0.5!(b_4)$) {$c_1$};
		\node[circle,draw,minimum size = 6.5*\nodedist] (c_2) at ($(a_2)!0.5!(b_5)$) {$c_2$};
		\node[circle,draw,minimum size = 6.5*\nodedist] (c_3) at ($(a_3)!0.5!(b_1)$) {$c_3$};
		\node[circle,draw,minimum size = 6.5*\nodedist] (c_4) at ($(a_4)!0.5!(b_2)$) {$c_4$};
		\node[circle,draw,minimum size = 6.5*\nodedist] (c_5) at ($(a_5)!0.5!(b_3)$) {$c_5$};

		\foreach \x in {1,...,10}
		{
			\node at ({360/\n * (\x)}:1.1*\radius) {$4$};
		}

		\path[every node/.style={font=\sffamily\small}]
		
		(a_1) edge [near end] node {$5$} (c_1)
		(a_2) edge [near end] node {$5$} (c_2)
		(a_3) edge [near end] node {$5$} (c_3)
		(a_4) edge [near end] node {$5$} (c_4)
		(a_5) edge [near end] node {$5$} (c_5)

		(b_1) edge [near end] node {$4$} (c_3)
		(b_2) edge [near end] node {$4$} (c_4)
		(b_3) edge [near end] node {$4$} (c_5)
		(b_4) edge [near end] node {$4$} (c_1)
		(b_5) edge [near end] node {$4$} (c_2)
		;

		\end{tikzpicture}
		\caption{A social symmetric FHG in which no core stable partition exists. The weight of an edge $\{i,j\}$ denotes $v_i(j)$. All missing edges have weight $0$.}
		\label{fig:socsymcore}
		\end{figure}

	The simple and symmetric FHG with empty core depicted in \figref{fig:simsymcorelabel} is derived from the game given in \figref{fig:socsymcore} by replacing all players $a_i$ and $c_i$ by a clique of $3$ players and all players $b_i$ by a clique of $2$ players.
	These cliques are denoted by $A_i$, $C_i$, and $B_i$, respectively.
	Then, the number of players a player is connected to in the union of two connected cliques in \figref{fig:simsymcorelabel} is equal to the weight of the edge between the corresponding players in \figref{fig:socsymcore}.

	\begin{theorem}
	In simple and symmetric FHGs, the core can be empty.
	\end{theorem}

	\begin{proof}
		The core of the FHG depicted in \figref{fig:simsymcorelabel} is empty.
		For two players $i,j\in N$ we say that $i$ is connected to $j$ if $i$'s valuation for $j$ is $1$ (and \emph{vice versa}).
		Let $\pi$ be in the core.
		The first step is to show that $A_l\subseteq S\in\pi$ and $C_l\subseteq T\in\pi$ for all $l\in\{1,\dots,5\}$.
		We show both statements for $l = 1$. 
		The rest follows from the symmetry of the game.  
	
		  	$A_1\subseteq S\in\pi$: Assume for contradiction that this is not the case. 
			Since $A_1\cup C_1$ is a $6$-clique, at least one player $i\in A_1\cup C_1$ has a valuations of at least $\nicefrac{5}{6}$ for his coalition (otherwise $A_1\cup C_1$ is blocking). 
			Assume $i\in A_1$. 
			If $\pi(i)$ contains a player that $i$ is not connected to, then $u_i(\pi) \leq \nicefrac{9}{11} < \nicefrac{5}{6}$ since $i$ is connected to at most $9$ players in any coalition. 
			Hence $\pi(i)$ only contains players $i$ is connected to.
			But then $A_1\cup \pi(i)$ is blocking, since every player in $A_1$ is connected to the same players as $i$, a contradiction.
			Hence $i\in C_1$. 
			$A_1\cap\pi(i) = \emptyset$ implies $u_i(\pi)\leq\nicefrac{4}{5}$. 
			If $\pi(i)$ contains a player that $i$ is not connected to, then $u_i(\pi)\leq\nicefrac{7}{9}<\nicefrac{5}{6}$, since $i$ is connected to at most $7$ players in any coalition.
			Hence, $\pi(i)\cap A_1=S\neq\emptyset$ and $\pi(i)$ only contains players $i$ is connected to. 
			Thus, $C_1\subseteq\pi(i)$ (otherwise $C_1\cup\pi(i)$ is blocking).
		
			At least one player $k_1$ in $A_1\cup B_1$ and at least one player $k_2$ in $A_1\cup B_5$ has a valuation of at least $\nicefrac{4}{5}$ for his coalition, since both sets are $5$-cliques. 
			$k_1,k_2\not\in S$, since $u_j(\pi)\leq\nicefrac{4}{6}$ for all $j\in S$.
			If $k_1\in A_1\setminus S$, then $\pi(k_1)$ only contains players that $k_1$ is connected to, otherwise $u_{k_1}(\pi)\leq\nicefrac{5}{7}<\nicefrac{4}{5}$. 
			Then $\pi(k_1)\cup S$ is blocking.
			Hence $k_1\in B_1$.
			Analogously, it follows that $k_2\in B_5$.
			We show that $\pi(k_1)\neq\pi(k_2)$.
			Assume for contradiction that $\pi(k_1) = \pi(k_2) = T$.
			If $S$ contains at least two players that $k_1$ is not connected to, we have $u_{k_1}(T)\leq\nicefrac{10}{13}<\nicefrac{4}{5}$ (since $k_1$ is connected to at most $10$ players in any coalition).
			Hence, $T$ contains one player $k_1$ is not connected to, namely $k_2$.
			The analogous assertion holds for $k_2$.
			Since $u_{k_2}(T)\geq\nicefrac{4}{5}$, we have $|T|\geq 10$.
			But then $T$ contains at least $2$ players $k_1$ is not connected to, since there are only $3$ players that both $k_1$ and $k_2$ are connected to.
			This implies that $u_{k_1}(T)\le\nicefrac{10}{13}<\nicefrac{4}{5}$, a contradiction.

			If $B_4\subseteq\pi(i)$, it follows that $u_j(\pi)\leq\nicefrac{4}{7}<\nicefrac{2}{3}$ for all $j\in S$.
			If $j\in A_1\setminus S$ is in a coalition with a player that $j$ is not connected to, $u_j(\pi)\leq\nicefrac{3}{5}<\nicefrac{2}{3}$, since $\pi(j)$ cannot contain a player from $\pi(i)$ and from both $\pi(k_1)$ and $\pi(k_2)$ (since $\pi(k_1)\neq\pi(k_2)$).
			Hence $S\cup\pi(j)$ is blocking.
		
			If $|\pi(i)\cap B_4| = 1$ it follows that $|S| = 2$ and $u_j(\pi) = \nicefrac{4}{6}$ for all $j\in S$.
			At least one player $k$ in $A_5\cup B_4$ has a valuation of at least $\nicefrac{4}{5}$ for his coalition.
			If $k\in\pi(i)$ it follows that $u_k(\pi) = \nicefrac{3}{6}<\nicefrac{4}{5}$, a contradiction.
			If $k\in B_4\setminus\pi(i)$, then $\pi(k)$ only contains players that $k$ is connected to.
			If $A_4\subseteq\pi(k)$ or $A_5\subseteq\pi(k)$, then $A_4\cup C_4$ and $A_5\cup C_5$ are blocking, respectively.
			$|\pi(k)\cap A_4| = 2$ or $|\pi(k)\cap A_5| = 2$ is not possible since our previous analysis for $A_1$ and $C_1$ also applies to $A_4$ and $C_4$, and  $A_5$ and $C_5$, respectively.
			But then, $u_k(\pi)\le\nicefrac{2}{3}<\nicefrac{4}{5}$.
			Hence $k\in A_5$. This implies that $\pi(k)$ only contains players that $k$ is connected to.
			Hence $A_5\subseteq\pi(k)$, otherwise $A_5\cup\pi(k)$ is blocking.
			Also $\pi(k)\neq A_5\cup B_5$, because otherwise $A_5\cup C_5$ is blocking.
			Hence $A_5\cap \pi(k_2) = \emptyset$. 
			Thus, $\pi(k_2)$ can only contain players $k_2$ is connected to.
			This implies $B_5\subseteq\pi(k_2)$.
			As $u_{k_2}(\pi)\geq\nicefrac{4}{5}$, $|\pi(k_2)\cap C_2|\geq 2$. 
			Thus, if $\pi(k_1)$ contains some player in $A_2$, then $A_2\cup C_2$ is blocking.
			If $\pi(k_2) = B_5\cup C_2$, then $A_2\cup C_2$ is blocking and otherwise $C_2$ is blocking.
			This contradicts that $\pi$ is stable.

			$C_1\subseteq T\in\pi$: At least one player $i$ in $B_4\cup C_1$ has a valuation of at least $\nicefrac{4}{5}$ for his coalition (otherwise $B_4\cup C_1$ is blocking).
			Assume $i\in B_4$ and $u_j(\pi) < \nicefrac{4}{5}$ for all $j\in C_1$. 
			Then $A_l\subseteq\pi(i)$ for some $l\in\{4,5\}$. 
			But then $A_l\cup C_l$ is blocking. 
			Hence the assumption is wrong and $i\in C_1$. 
			Note that $\pi(i)$ cannot contain a player $i$ is not connected to, otherwise $u_i(\pi)\leq\nicefrac{7}{9}<\nicefrac{4}{5}$, since $i$ is connected to at most $7$ players in any coalition.
			But then $C_1\subseteq\pi(i)$, otherwise $C_1\cup\pi(i)$ is blocking. 	

	It cannot be that $\pi(i)\subseteq A_l\cup B_l$ or $\pi(i)\subseteq A_l\cup B_{l-1}$ for $i\in A_l$, since $A_l\cup C_l$ is blocking for all $l\in\{1,\dots,5\}$.

	If $A_1\cup C_1\cup S\in\pi$ with $\emptyset\neq S\subseteq B_4$, then $u_i(\pi)\le\nicefrac{5}{7}<\nicefrac{4}{5}$ for all $i\in A_1$.
	Hence $B_1\cup C_3, B_5\cup C_2\in\pi$, otherwise either $A_1\cup B_1$ or $A_1\cup B_5$ are blocking ($u_i(A_1\cup B_1) = u_i(A_1\cup B_5) = \nicefrac{4}{5}$ for all $i\in A_1$). 
	But then $u_i(\pi)\le\nicefrac{4}{5}$ for all $i\in A_2$.
	Hence $A_2\cup C_2$ is blocking, a contradiction.
	In any other partition in which some $i\in A_1$ is in a coalition with a player that he is not connected to, we have $u_i(\pi)\le\nicefrac{9}{11}<\nicefrac{5}{6}$ for all $i\in A_1$ and $u_j(\pi)\le\nicefrac{4}{5}<\nicefrac{5}{6}$ for all $j\in C_1$.
	Hence $A_1\cup C_1$ is blocking.
	Hence $\pi(i)$ only contains players $i$ is connected to for all $i\in A_l$ and $l\in\{1,\dots,5\}$.

	We have shown previously that at least one player $i_l\in C_l$ has a valuation of at least $\nicefrac{4}{5}$ for his coalition for all $l\in\{1,\dots,5\}$. 
	Hence, $\pi(i_l)$ cannot contain a player that $i_l$ is not connected to.
	Therefore, either $\pi(i_l) = A_l\cup C_l$ or $\pi(i_l) = B_{l-2}\cup C_l$ for all $l\in\{1,\dots,5\}$. 
	If $A_l\cup C_l\in\pi$ for all $j\in\{1,\dots,5\}$, then $A_1\cup B_1\cup B_5$ is blocking.
	Hence we can assume without loss of generality that $A_1\cup S\in\pi$ with $S\subseteq B_1\cup B_5$.
	If $|S|<3$, then $A_1\cup C_1$ is blocking.
	Hence $|S|\geq 3$.
	Without loss of generality, $B_5\subseteq S$.
	It follows that $B_4\cup C_1\in\pi$, since one player in $C_1$ has a valuation of at least $\nicefrac{4}{5}$ for his coalition.
	This implies that $A_4\cup C_4\in\pi$.
	Furthermore $A_2\cup C_2, A_3\cup C_3\in\pi$, otherwise $B_5\cup C_2$ and $B_1\cup C_3$ are blocking, respectively.
	Then we get $A_5\cup C_5\in\pi$.
	But then $A_3\cup B_2\cup B_3$ is blocking.
	Hence, $\pi$ is not in the core, a contradiction.

	\begin{figure}
		\centering
			\begin{tikzpicture}[-, line width = 0.5pt]
		
			\tikzset{
			    position/.style args={#1:#2 from #3}{
			        at=(#3.#1), anchor=#1+180, shift=(#1:#2)
			    }
			}


			\def \n {10}
			\def \k {9}
			\def \radius {18ex}
			\def \nodedist {.9ex}
			\def \margin {8} 
			\def \cliqueopacity {1} 
			\def \labelposition {1.25} 
			\def \labelpositionC {.85}
	
			\draw (0,0) circle (\radius);

			\node[circle,draw,minimum size = 6.5*\nodedist,fill = white] (a_1) at ({360/\n * (2.5)}:\radius) {};
			\node[circle,draw,minimum size = 6.5*\nodedist,fill = white] (a_2) at ({360/\n * (0.5)}:\radius) {};
			\node[circle,draw,minimum size = 6.5*\nodedist,fill = white] (a_3) at ({360/\n * (8.5)}:\radius) {};
			\node[circle,draw,minimum size = 6.5*\nodedist,fill = white] (a_4) at ({360/\n * (6.5)}:\radius) {};
			\node[circle,draw,minimum size = 6.5*\nodedist,fill = white] (a_5) at ({360/\n * (4.5)}:\radius) {};
	
			\node[circle,draw,minimum size = 6.5*\nodedist,fill = white] (b_1) at ({360/\n * (1.5)}:\radius) {};
			\node[circle,draw,minimum size = 6.5*\nodedist,fill = white] (b_2) at ({360/\n * (9.5)}:\radius) {};
			\node[circle,draw,minimum size = 6.5*\nodedist,fill = white] (b_3) at ({360/\n * (7.5)}:\radius) {};
			\node[circle,draw,minimum size = 6.5*\nodedist,fill = white] (b_4) at ({360/\n * (5.5)}:\radius) {};
			\node[circle,draw,minimum size = 6.5*\nodedist,fill = white] (b_5) at ({360/\n * (3.5)}:\radius) {};
	
			\node[circle,draw,minimum size = 6.5*\nodedist] (c_1) at ($(a_1)!0.5!(b_4)$) {};
			\node[circle,draw,minimum size = 6.5*\nodedist] (c_2) at ($(a_2)!0.5!(b_5)$) {};
			\node[circle,draw,minimum size = 6.5*\nodedist] (c_3) at ($(a_3)!0.5!(b_1)$) {};
			\node[circle,draw,minimum size = 6.5*\nodedist] (c_4) at ($(a_4)!0.5!(b_2)$) {};
			\node[circle,draw,minimum size = 6.5*\nodedist] (c_5) at ($(a_5)!0.5!(b_3)$) {};

			\node at ({360/\n * (2.5)}:\radius*\labelposition) {$A_1$};
			\node at ({360/\n * (0.5)}:\radius*\labelposition) {$A_2$};
			\node at ({360/\n * (8.5)}:\radius*\labelposition) {$A_3$};
			\node at ({360/\n * (6.5)}:\radius*\labelposition) {$A_4$};
			\node at ({360/\n * (4.5)}:\radius*\labelposition) {$A_5$};	
			
			\node at ({360/\n * (1.5)}:\radius*\labelposition) {$B_1$};
			\node at ({360/\n * (9.5)}:\radius*\labelposition) {$B_2$};
			\node at ({360/\n * (7.5)}:\radius*\labelposition) {$B_3$};
			\node at ({360/\n * (5.5)}:\radius*\labelposition) {$B_4$};
			\node at ({360/\n * (3.5)}:\radius*\labelposition) {$B_5$};
			
			\node at ({360/\n * (2)}:\radius*\labelpositionC) {$C_2$};
			\node at ({360/\n * (10)}:\radius*\labelpositionC) {$C_3$};
			\node at ({360/\n * (8)}:\radius*\labelpositionC) {$C_4$};
			\node at ({360/\n * (6)}:\radius*\labelpositionC) {$C_5$};
			\node at ({360/\n * (4)}:\radius*\labelpositionC) {$C_1$};

			\foreach \x in {1,2,3,4,5}
			{
	
				\foreach \y in {1,2,3}
				{
					\node[circle,draw,position=42+120*\y-72*\x:{-2.7*\nodedist} from a_\x,opacity = \cliqueopacity] (a_\x\y) {};
					\node[circle,draw,position=54+42+120*\y-72*\x:{-2.7*\nodedist} from c_\x,opacity = \cliqueopacity] (c_\x\y) {};
			
				}
		
				\node[circle,draw,position=0-36+72-72*\x:{-2.7*\nodedist} from b_\x,opacity = \cliqueopacity] (b_\x1) {};
				\node[circle,draw,position=180-36+72-72*\x:{-2.7*\nodedist} from b_\x,opacity = \cliqueopacity] (b_\x2) {};

				\draw[opacity = \cliqueopacity] (a_\x1) -- (a_\x2);
				\draw[opacity = \cliqueopacity] (a_\x1) -- (a_\x3);	
				\draw[opacity = \cliqueopacity] (a_\x3) -- (a_\x2);
		
				\draw[opacity = \cliqueopacity] (c_\x1) -- (c_\x2);
				\draw[opacity = \cliqueopacity] (c_\x1) -- (c_\x3);	
				\draw[opacity = \cliqueopacity] (c_\x3) -- (c_\x2);
		
				\draw[opacity = \cliqueopacity] (b_\x1) -- (b_\x2);
		
			}
  
			\path[every node/.style={font=\sffamily\small}]
	
			(a_1) edge [near end] node {} (c_1) 
			(a_2) edge [near end] node {} (c_2) 
			(a_3) edge [near end] node {} (c_3) 
			(a_4) edge [near end] node {} (c_4) 
			(a_5) edge [near end] node {} (c_5) 
	
			(b_1) edge [near end] node {} (c_3) 
			(b_2) edge [near end] node {} (c_4) 
			(b_3) edge [near end] node {} (c_5) 
			(b_4) edge [near end] node {} (c_1) 
			(b_5) edge [near end] node {} (c_2) 
			;

			\end{tikzpicture}
		\caption{A simple and symmetric FHG with empty core. For all $l\in\{1,\dots,5\}$, $A_l$ and $C_l$ denote cliques of $3$ players and $B_l$ denotes a clique of $2$ players. An edge from one clique to another denotes that every player in the first clique is connected to every player in the second clique. All depicted edges have weight $1$. All missing edges have weight $0$.}
		\label{fig:simsymcorelabel}
		\end{figure}
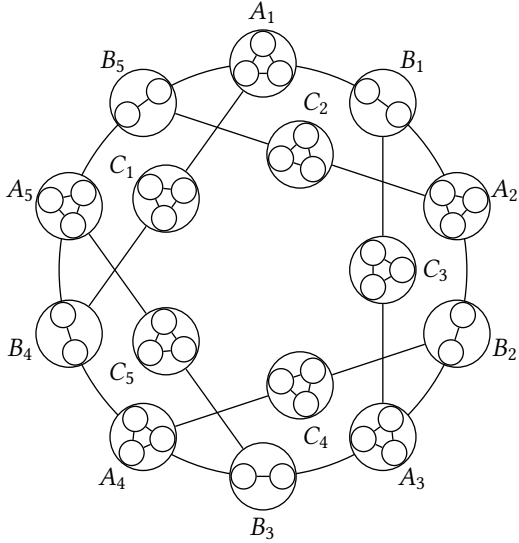
		\end{proof}
		
		\begin{remark}
			\label{rmk:stable-sub-example}
			Suppose we delete one of the players from $B_2$ in the game above, so that now $|B_2| = 1$. The resulting game admits a partition $\pi$ in its core: $\pi = \{ A_1 \cup B_1 \cup B_5, A_2 \cup C_2, A_3 \cup C_3, A_4 \cup C_4, A_5 \cup C_5, B_4 \cup C_1, B_2, B_3 \}$. That $\pi$ is in the core can be checked by hand or by computer.
		\end{remark}
		
		\subsection{Hardness results}
		
		We will now show that it is computationally hard to decide whether a given FHG admits a non-empty core. This problem turns out to be $\Sigma_2^p$-complete, and thus complete for the second level of the polynomial hierarchy, even for simple and symmetric FHGs. Our argument is rather involved; shorter proofs exist when aiming only for NP-hardness and without the restriction to simple and symmetric games \citep{BBS14a,PeEl15a}.
		
		We will start our reduction from the problem MAXMIN-CLIQUE, which is $\Pi_2^p$-complete \citep{KoLi95a}:
		
		\begin{quote}
		\textbf{MINMAX-CLIQUE} \\
		\textbf{Instance:} An undirected graph $H = (V, E)$ whose vertex set $V = \bigcup_{i = 1}^n \bigcup_{j = 1}^c V_{i,j}$ is partitioned into a grid with $n$ rows and $c$ columns, and a target integer $k$.\\
		\textbf{Question:} Is it the case that for every way of choosing exactly one $V_{i,j}$ for each row $i$, the union of the $n$ chosen cells contains clique of size $k$?
		\end{quote}
		
		From the reduction presented by \citet{KoLi95a}, it follows that this problem remains $\Pi_2^p$-complete even if $c = 2$, all the $V_{i,j}$'s contain the same number of vertices (say $|V_{i,j}| = m$), and $k = n$.
		From this, it is easy to see that the problem with $k = n+\frac{nm}2$ is also hard: just add a clique of $2nm$ ($=|V|$) new vertices to $H$, connect each of the new vertices to every of the old vertices, and distribute the new vertices into the grid so that each $V_{i,j}$ contains precisely $m$ of these new vertices. Note that after this reduction, the value of $m$ has doubled.
		
		Taking the complement of the problem we have now arrived at, we find that the following problem is $\Sigma_2^p$-complete. (Note the change from ``minmax'' to ``maxmin''.)
		
		\begin{quote}
		\textbf{MAXMIN-CLIQUE} \\
		\textbf{Instance:} An undirected graph $H = (V, E)$ whose vertex set $V = \bigcup_{i = 1}^n \bigcup_{j = 1,2} V_{i,j}$ is partitioned into a grid with $n$ rows and 2 columns, where all cells contain the same number of vertices, say $|V_{i,j}| = m$ for all $i,j$. \\
		\textbf{Question:} Is there a way to choose exactly one of $V_{i,1}$ and $V_{i,2}$ for each row $i$ so that the union of the $n$ chosen cells does \emph{not} contain a clique of size $n+\frac{nm}2$?
		\end{quote}
		
		We will not give a direct reduction from MAXMIN-CLIQUE to our problem about FHGs, but will instead consider an intermediate problem first. Later, we show how to extend this to the case we are actually interested in. Our intermediate problem uses a modification of allowing so-called \emph{supported players}, who are unusually happy in a singleton coalition. A similar device also appears in the $\Sigma_2^p$-hardness proof by \cite{Pete15a} for additively separable hedonic games. The formal definition of our problem is as follows.
		
		\begin{quote}
			\textbf{Core-non-emptiness with Supported Players} \\
		\textbf{Instance:} An undirected unweighted graph $G = (N,E)$, defining an FHG. This hedonic game is then modified by identifying a number of \textit{supported} players $S\subseteq N$ who receive a specified subsidy when they are alone, i.e., for each $i\in S$, we set $v(\{i\}) = (l_i - 1)/l_i$ for some given integer $l_i \ge 4$ (encoded in unary). \\
		\textbf{Question:} Does the given hedonic game admit a non-empty core?
		\end{quote}
		
		Later we will show a reduction from this problem to the case without supported players; there the technical assumption that the subsidies satisfy $l_i \ge 4$ will become useful.
		
		\begin{theorem}
			Core-non-emptiness with Supported Players is $\Sigma_2^p$-complete.
		\end{theorem}
		
		\begin{proof}
			We reduce from MAXMIN-CLIQUE. So let $H = (V,E)$ be a given graph with vertex partition $V = \bigcup_{i = 1}^n \bigcup_{j = 1,2} V_{i,j}$ with $|V_{i,j}| = m$ for all $i$ and $j$ and with target clique size $k = n+\frac{nm}2$. We now construct a game $G = (N,E')$ with supported players $S$.
			
			Let $M$ be a big number, $M = 20m^2n$ will do.
			
			We produce the following players.
			\begin{itemize}
				\item For each row $i$, we introduce a player $z_i$ who will eventually be responsible for choosing one of the cells $V_{i,1}$ or $V_{i,2}$.
				\item For each of the two cells $V_{i,1}$ and $V_{i,2}$ of a row, we introduce a set of $M$ supported players; $|X_{i,1}| = |X_{i,2}| = M$. Each of these players receive subsidy $(M+2m)/(M+2m+1)$.
				\item Each original vertex $v\in V$ is also a player $v\in N$.
				\item For each original vertex $v\in V$, we introduce a \textit{mate} $v'$.
				\item For each $v\in V$, we also introduce a set $C_v$ of $k-3$ supported players with subsidy $(k-2)/(k-1)$.
				\item For each player $z_i$, we introduce a set $O_{z_i}$ of 39 players who will form a copy of the game from Theorem~\ref{thm:simsymcore} with empty core.
				\item For each mate player $v'$, we introduce a set $O_{v'}$ of 39 players who will form a copy of the game from Theorem~\ref{thm:simsymcore} with empty core.
			\end{itemize}
			
			If $W \subseteq V$ is a subset of vertices, let's write $W' = \{ v' : v\in W \}$ for the collection of mates of vertices in $W$.
			Summarizing, we have produced the following set of players:
			\[ N = V \cup V' \cup \{ z_i : i \in [n] \} \cup \bigcup_{\mathclap{i,j \in [n] \times [2]}} X_{i,j} \cup \bigcup_{v\in V} C_v \cup \bigcup_{v' \in V'} O_{v'} \cup \bigcup_{i \in [n]} O_{z_i} , \]
			of which the following are supported:
			\[ S = \bigcup_{\mathclap{i,j \in [n] \times [2]}} X_{i,j} \cup \bigcup_{v\in V} C_v. \]
			
			We also need to construct the set of edges $E'$:
			
			\begin{itemize}
				\item All original edges from $E$ are in $E'$.
				\item For each $v$, the set $C_v \cup \{v,v'\}$ forms a clique of size $k-1$.
				\item For each cell $V_{i,j}$, the set $X_{i,j} \cup \{ z_i \}$ forms a clique of size $M + 1$.
				\item For each $v\in V_{i,j}$, both $v$ and $v'$ are connected to all vertices in $X_{i,j}$.
				\item The sets $O_{z_i} \cup \{z_i\}$ and $O_{v_i'} \cup \{v_i'\}$ form a copy of the game from Theorem~\ref{thm:simsymcore}, such that the distinguished player ($z_i$ and $v_i$ respectively) is one of the two players in $B_2$.
				\item There are no other edges.
			\end{itemize}
			
			This completes the description of the reduction. \\
			
			$\implies$: Suppose the game $G$ admits a core-stable partition $\pi$. We show how to choose cells $t  : [n] \to \{1,2\}$ so that $\bigcup_i V_{i, t(i)}$ contains no clique of size $k$.
			
			Note first that every coalition in $\pi$ needs to be connected in $G$, since otherwise a connected component blocks $\pi$. Consider some row $i\in [n]$. Since the game restricted to the players in $O_{z_i} \cup \{z_i\}$ does not possess a core-stable partition, and by connectedness, the player $z_i$ needs to be together with one of his neighbors outside $O_{z_i}$, i.e., with a neighbor from $X_{i,1} \cup X_{i,2}$. Say this neighbor $x$ comes from $X_{i,1}$, so $x \in \pi(z_i) \cap X_{i,1}$. We show that in fact $\pi(z_i) = \{z_i\} \cup X_{i,1} \cup V_{i,1}' \cup V_{i,1}$.
			
			\begin{itemize}
				\item $\supseteq$: We know that $x\in \pi(z_i)$. Now $x$ is supported with subsidy $(M+2m)/(M+2m+1)$; since $\{x\}$ does not block $\pi$, it must be the case that $x$'s utility in $\pi$ is at least as high as its subsidy. Hence $|\pi(z_i)| = |\pi(x)| \ge M+2m+1$, and $x$ must have at least $M+2m$ neighbors in $\pi(z_i)$. Recalling that $m=|V_{i,j}|$ for all $i$ and $j$, and looking at the reduction, we see that $x$ only has $M+2m$ neighbors in total, namely $V_{i,1} \cup V_{i,1}' \cup \{z_i\} \cup X_{i,1} \setminus \{x\}$, and hence this must form a subset of $\pi(z_i)$.
				\item $\subseteq$: If there are any additional players in $\pi(z_i)$, then $x$ obtains utility strictly less than $(M+2m-1)/(M+2m)$, and then $\{x\}$ would block $\pi$, invoking his subsidy.
			\end{itemize}
			
			We deduce that for each row $i$, either $\pi(z_i) = \{z_i\} \cup X_{i,1} \cup V_{i,1}' \cup V_{i,1}$ or $\pi(z_i) = \{z_i\} \cup X_{i,2} \cup V_{i,2}' \cup V_{i,2}$. 
			This allows us to choose cell $V_{i,2}$ in the former case (setting $t(i)=2$) and $V_{i,1}$ in the latter (setting $t(i)=1$). Note that $z_i$ is together with the cell that is \emph{not} chosen.
			
			Now let's consider the players in $X_{i,t(i)}$ corresponding to a chosen cell. Given what we know so far about $\pi$, these players only have $M+2m-1$ remaining neighbors (since $z_i$ is in a different coalition). Thus, no non-singleton coalition can give such a player utility exceeding the subsidy $(M+2m)/(M+2m+1)$. Hence each player in $X_{i,t(i)}$ is in a singleton in $\pi$.
			
			Now consider a vertex $v\in V_{i,t(i)}$ in a chosen cell and look at its mate $v'$. Since the game restricted to the players in $O_{v'} \cup \{v'\}$ does not possess a core-stable partition, the player $v'$ needs to be together with a neighbor outside $O_{v'}$, i.e., needs to be together with $v$ and/or a player in $C_v$. In fact, we can see that $v'$ needs to be together with at least one player from $C_v$: if not, then $v'$ obtains utility at most $11/12$ (because $v'$ has 10 neighbors in $O_{v'}$ plus the neighbor $v$), and then $C_v \cup \{v'\}$ is blocking. Thus we have shown that there is $c\in C_v$ with $c \in \pi(v')$. Since $\{c\}$ is not blocking, $c$'s utility in $\pi$ must exceed its subsidy $(k-2)/(k-1)$. Thus $|\pi(v')| = |\pi(c)| \ge k-1$ and $c$ needs to have at least $k-2$ neighbors in $\pi(v')$. But $c$ has exactly $k-2$ neighbors, and so, like above, we have $\pi(v') = \{v,v'\} \cup C_v$. In particular, each $v\in V_{i,t(i)}$ in a chosen cell obtains utility $(k-2)/(k-1)$.
			
			Finally, suppose for a contradiction that the union $\bigcup_i V_{i,t(i)}$ of the chosen cells contains a clique $K\subseteq V$ of size $k$. Then each vertex of $K$ obtains utility $(k-1)/k$ in $K$, so $K$ blocks $\pi$, a contradiction. Hence, with our choice of $t : [n] \to \{1,2\}$, the set $\bigcup_i V_{i,t(i)}$ contains no clique of size $k$.  \\
			
			$\impliedby$: Suppose there is a way of choosing $t : [n] \to \{1,2\}$ so that $\bigcup_i V_{i,t(i)}$ contains no clique of size $k$. We construct a partition $\pi$ of $N$ which is core-stable. For each row $i\in [n]$:
			
			\begin{itemize}
				\item $\{x\}\in \pi$ for each $x\in X_{i,t(i)}$.
				\item $\{v,v'\} \cup C_v \in \pi$ for each $v\in V_{i,t(i)}$.
				\item $\{z_i\} \cup X_{i,\lnot t(i)} \cup V_{i,\lnot t(i)}' \cup V_{i,\lnot t(i)} \in \pi$, where $\lnot t(i) = 3 - t(i)$ is the not-chosen index.
				\item $\{c\} \in \pi$ for each $c \in C_v$ for $v\in V_{i,\lnot t(i)}$.
				\item The players in sets $O_{z'}$ and $O_{v'}$ are partitioned in the way indicated in Remark~\ref{rmk:stable-sub-example}.
			\end{itemize}
			
			Let us note first that each player $z_i$ receives utility $M/(M+2m+1) > 11/12$, and also each mate $v'$ either receives utility $(k-2)/(k-1) > 11/12$ or $(1 + M)/(M+2m+1) > 11/12$ since $M$ is chosen large enough. Because $z_i$ and $v'$ players only have 10 neighbors in $O_{z_i}$ and $O_{v'}$, respectively, they will not block in a coalition that is contained entirely within $O_{z_i} \cup \{z_i\}$ or $O_{v_i} \cup \{v_i\}$, because those only bring utility at most $10/11$.
			
			We now show that $\pi$ admits no blocking coalitions. To do so, we will go through all the players to check that they have no incentive to deviate. (We will say that a player $i$ is \emph{not blocking} if $i$ is not part of any blocking coalition.) First notice that $\pi$ is individually rational, and in particular every supported player receives at least its subsidy. Therefore no singleton coalition blocks $\pi$.
			
			\begin{itemize}
				\item The players $x \in X_{i,\lnot t(i)}$ are in a best-possible coalitions: they are together with exactly their neighbors. Hence they will never be part of a blocking coalition.
				\item The players $x \in X_{i,t(i)}$ (who form singletons in $\pi$ and currently receive their subsidy) will not deviate, because the only coalition that exceeds their subsidy would be $x$'s neighborhood $\{z_i\} \cup X_{i,t(i)} \cup V_{i, t(i)}' \cup V_{i,t(i)}$, yet this is not a blocking coalition since $z_i$ is not better off in it. 
				\item For each $z_i$, we have excluded all the neighbors $x\in X_{i,j}$ of $z_i$ as possible members of a blocking coalition. This would only leave a blocking coalition contained entirely within $O_{z_i} \cup \{z_i\}$, which is not an improvement for $z_i$ as argued above. Hence $z_i$ will not block.
				\item Each $c \in C_v$ where $v\in V_{i,t(i)}$ is in a best-possible coalition because its coalition is precisely its neighborhood, and hence will not deviate.
				\item Each $c \in C_v$ where $v\in V_{i,\lnot t(i)}$ (who forms a singleton in $\pi$ and currently receives its subsidy) cannot block, because the only coalition that exceeds its subsidy would be its neighborhood $\{v,v'\} \cup C_v$, but this coalition is not blocking because $v'$ is not better off (since $(1 + M)/(M+2m+1) > (k-2)/(k-1)$ by choice of $M$ large enough).
				\item Consider a mate player $v'$. We have already shown that all of its neighbors, except possibly $v$ and from $O_{v'}$, are not part of blocking coalitions. But a blocking coalition contained in $\{v,v'\} \cup O_{v'}$ brings utility at most $11/12$ to $v'$ (because $v'$ only has 11 neighbors in this set), and thus $v'$ is not better off in such a coalition. Hence no mate player is blocking.
				\item No player in $O_{z_i}$ or $O_{v'}$ can be part of a blocking coalition by Remark~\ref{rmk:stable-sub-example}, since $z_i$ and $v'$, respectively, are not blocking.
				\item Each $v\in V_{i,\lnot t(i)}$ currently receives utility $\ge (1 + M)/(M+2m+1)$ which, for our choice of $M$ large enough, exceeds the utility $v$ could receive in any coalition $S\subseteq V$ consisting entirely of original vertices (this quantity being at most $(|V|-1)/|V| = (2nm-1)/2nm$).
				\item Thus, any blocking coalition $S$ that we have not yet excluded must consist entirely of original vertices in chosen cells, that is $S\subseteq \bigcup_{i=1}^n V_{i, t(i)}$. Because each $v\in S$ currently obtains utility $(k-2)/(k-1)$, $S$ must give each member a utility exceeding this value. We show that $S$ is a clique in the graph $H$, and of size $\ge k$, which gives a contradiction.
				
				Let $r := |S|$. Note that $r \le mn = |\bigcup_{i=1}^n V_{i, t(i)}|$. Suppose that $S$ is not a clique. Then there exists a vertex $v\in S$ which is not connected to every other $w\in S$. Thus 
				\[ u_v(S) \le \frac{r-2}{r} \le \frac{mn-2}{mn} < \frac{k-2}{k-1},  \]
				where the last inequality follows from $k > \frac{mn}{2} + 1$ by simple algebra. But because $S$ is assumed to be blocking, we know that $(k-2)/(k-1) < u_v(S)$, and hence we have a contradiction. Thus, $S$ must be a clique. Since each $v\in S$ obtains utility $> (k-2)/(k-1)$ in it, we must have that $|S| \ge k$, a contradiction.
			\end{itemize}
			Thus, no blocking coalition exists, and hence $\pi$ is in the core.
		\end{proof}
		
		With this result in place, we can now formally state the following problem, and prove it to be hard:
		\begin{quote}
			\textbf{Core-non-emptiness for simple and symmetric FHGs} \\
			\textbf{Instance:} An undirected unweighted graph $G = (N,E)$, defining an FHG. \\
			\textbf{Question:} Does the given FHG game admit a non-empty core?
		\end{quote}
		
		\begin{theorem}
			\label{thm:reduction-supported}
			Core-non-emptiness for simple and symmetric FHGs is $\Sigma_2^p$-complete.
		\end{theorem}
		\begin{proof}
			We reduce from Core-non-emptiness with Supported Players. So let $G = (N,E)$ be a FHG modified by having supported players $S \subseteq N$ where $i\in S$ gets subsidy $(l_i - 1)/l_i$ where $l_i \ge 4$.
			
			We build a new FHG $H = (N', E')$ without supported players such that $G$ possesses a core-stable partition if and only if $H$ does.
			
			The player set $N'$ of $H$ subsumes every original player from $N$, so $N \subseteq N'$. In addition, for each supported player $i\in S$, we add a set $C_{i}$ of $(l_i - 1)$ new players. Together we have $N' = N \cup \bigcup_{i\in S} C_i$.
			
			The edge set $E'$ of $H$ subsumes the original edges, so $E \subseteq E'$. Also, the sets $C_i \cup \{i\}$ form a clique of $l_i$ players for each $i\in S$. There are no other edges.
			
			This completes the description of the reduction. Before we prove correctness, let us analyze the preferences of players $j\in C_i$. Clearly, $j$'s unique most-preferred coalition is $C_i \cup \{i\}$, which is precisely $j$'s neighborhood. Ranked second are all coalitions of $(l_i-2)$ neighbors of $j$ together with $j$ (that is, the coalitions $C_i$ and $C_i \setminus \{k\} \cup \{i\}$ for some $k \in C_i \setminus \{j\}$) which give $j$ utility $(l_i - 2)/(l_i - 1)$. All other coalitions are ranked lower than these: let $C\ni j$ be any other coalition. 
			\begin{itemize}
				\item If $|C| \le l_i -2$, then $j$ obtains utility $\le (l_i - 3)/(l_i-2) < (l_i - 2)/(l_i - 1)$.
				\item If $|C| = l_i-1$, then, because $C$ is a coalition different from the ones considered above, $j$ has at most $l_i -3$ neighbors in $C$, so obtains utility $\le (l_i - 3)/(l_i-1) < (l_i - 2)/(l_i - 1)$.
				\item If $|C| = l_i$, then again, because $C$ is assumed to be different from $C_i \cup \{i\}$, $j$ has at most $l_i -2$ neighbors in $C$, so obtains utility $\le (l_i - 2)/l_i < (l_i - 2)/(l_i - 1)$.
				\item If $|C| \ge l_i + 1$, then $j$ can obtain utility at most $(l_i - 1)/(l_i + 1)$, which is worse than $(l_i - 2)/(l_i - 1)$ for $l_i > 3$.
			\end{itemize}
			Suppose $G$ has a core-stable partition $\pi$. Consider the following partition $\pi'$ of $H$: 
			\[\pi' = (\pi \setminus \{ \{i\} : i \in S \}) \cup \{ C_i \cup \{i\} : i\in S \text{ and } \{i\} \in \pi \} \cup \{ C_i :  i\in S \text{ and } \{i\} \not\in \pi \}.\]
			Thus, supported players $i\in S$ who are in a singleton coalition in $\pi$ join the coalition $C_i \cup \{i\}$ in $\pi'$.
			
			 We claim that $\pi'$ is core-stable in $H$. Note first that sets of form $C_i \cup \{i\}$ are not blocking: In this coalition, player $i$ receives utility $(l_i - 1)/l_i$ (which is equal to $i$'s subsidy) and so if $C_i \cup \{i\}$ was blocking $\pi'$, then $\{i\}$ would be blocking $\pi$. As we have seen, the coalitions $C_i$ are ranked second-best for its members, who therefore do not block either. Hence no player from any $C_i$ is blocking. Hence any potential blocking coalition for $\pi'$ is contained entirely in $N'$, and hence would also be a blocking coalition for $\pi$, which is a contradiction. Hence $\pi'$ is core-stable.
			
			Suppose $H$ has a core-stable partition $\pi'$. First note that for each $i\in S$, either $C_i \cup \{i\} \in \pi'$ or $C_i\in \pi'$, since otherwise either $C_i$ or $C_i \cup \{i\}$ blocks (by our observations about the preferences of players $j\in C_i$ above). Build the following partition $\pi$ of $N$: if $C_i \cup \{i\} \in \pi'$ then put $i$ in a singleton in $\pi$: $\{i\}\in \pi$; and for every $S\in \pi'$ with $S\subseteq N$, also put $S \in \pi$. The result is core-stable in $G$: for suppose not, and there is a blocking coalition $S\subseteq N$. If $S = \{i\}$ is a singleton and $i$ is supported, then $C_i \cup \{i\}$ would block $\pi'$. In all other cases $S$ would also block $\pi'$. Both cases give a contradiction, and so $\pi$ is core-stable.
		\end{proof}
	
		The reduction in the proof of Theorem~\ref{thm:reduction-supported} can also be used to show that it is coNP-complete to verify whether a given coalition structure is core-stable in a given simple symmetric FHG. We only sketch the argument, which is by reduction from clique. Given an instance $(G,k)$ of the clique problem (which asks whether $G$ contains a clique of size at least $k$, where we may assume that $k\ge \frac n2 +2$), we produce an FHG based on the same graph $G$, and make every vertex a supported player with subsidy $(k-2)/(k-1)$. In this game with supported players, the all-singletons coalition structure is in the core, unless there is a clique in $G$ of size at least $k$ (whence the clique would block, giving its members the payoff $(k-1)/k$). The approach of Theorem~\ref{thm:reduction-supported} can then be used to get rid of the subsidies.
		
		\citet{LiWe17a} give an alternative hardness proof of this problem. They show that verifying whether the grand coalition is core stable is coNP-complete even for graphs of degree 2 and that satisfy some further structural constraints, also by a reduction from the clique problem. \citet{LiWe17a} also present some heuristic approaches to answer this question.

		\subsection{Positive results}
		
		\subsection*{Graphs with bounded degree}
		
		\begin{theorem}
			For simple and symmetric FHGs represented by graphs of degree at most 2, the core is non-empty. 	
		\end{theorem}
		\begin{proof}
		
		We present a polynomial-time algorithm to compute a partition in the core. The partition is computed as follows. First keep finding $K_3$s until no more can be found. This takes time ${n \choose 3}$. Let us call the set of vertices matched into $K_3$s as $V_3$. We remove $V_3$ from the graph along with $E_3$---the edges between vertices in $V_3$. We then repeat the procedure by deleting $K_2$s instead of $K_3$s. Let us call the set of vertices matched into pairs by $V_2$. In that case, $V\setminus (V_2\cup V_3)$ are the unmatched vertices. The partition obtained is $\pi$. 
	
		In order to prove that $\pi$ is in the core, consider the potential blocking coalitions. We know that vertices in $V_3$ cannot be in a blocking coalition because each vertex in $V_3$ is in its most favored coalition. Also there does not exist a blocking coalition consisting solely of vertices from $V\setminus (V_2\cup V_3)$. If this were the case, then we had not deleted all $K_2$s from $(V\setminus V_3, E\setminus E_3)$. Now let us assume that there exists a $v_2\in V_2$ which is in a blocking coalition.
		A blocking coalition has to be of size $3$, since $v_2$ has utility $\nicefrac{1}{2}$ in $\pi$ and utility at most $\nicefrac{1}{2}$ in any coalition of size $2$ or size at least $4$. Moreover, a blocking coalition cannot contain two vertices from $V_1$, since for this to be the case $v_2$ has to be connected to one vertex in $V_2$ and two vertices in $V_1$, which violates the degree constraints.
		Hence, the coalition is of the form $\{v_1, v_2, v_2'\}$ where $v_1\in V\setminus (V_2\cup V_3)$ and $v_2, v_2'\in V_2$. If the utility of $v_2$ is greater than $\nicefrac{1}{2}$, then the utility of $v_2'$ is less than $\nicefrac{1}{2}$. Since $v_2'$ obtained utility $\nicefrac{1}{2}$ in $\pi$, $\{v_1, v_2, v_2'\}$ is not a blocking coalition.\end{proof}

		\subsection*{Forests}
		
		\begin{theorem}
			For simple and symmetric FHGs represented by undirected forests, the core is non-empty.
		\end{theorem}
		\begin{proof}	
			We present an algorithm to compute a partition in the core for an undirected tree. We can assume that the graph is connected---and therefore a tree---because if it were not, then the same algorithm for a tree could be applied to each connected component separately.
		Pick an arbitrary vertex $v_0\in V$ and run breadth first search on it.
		Let $L_0$ consist of $v_0$, $L_1$ of all the vertices at a distance of~$1$ from $v_0$, and~$L_k$ of all vertices at a distance of~$k$ from~$v_0$. 
		Let $L_\ell$ be the last layer of the tree.
		We construct a partition~$\pi$, which we will later claim is in the core. 
		Initialize~$\pi$ to the empty set. For each vertex~$v$ in the second last layer $L_{\ell-1}$ which has a child in the last layer $L_\ell$, add the set $\{v\}\cup \{w\midd w\in L_{\ell} \textrm{~and~} \{v,w\}\in E\}$ to $\pi$. Remove the sets of this form $\{v\}\cup \{w\midd w\in L_{\ell} \textrm{~and~} \{v,w\}\in E\}$ from the tree and repeat the process until no more layers are left.
		The partition returned is $\pi$. The procedure terminates properly. In each iteration, the last layer of the tree is removed along with some or all the vertices of the second last layer. If a vertex is left alone, send it to a smallest coalition that one of its neighbors is a member of.

		We now prove that~$\pi$ is in the core. For the base case, we show that no vertex from a coalition~$S\in\pi$ consisting of only the lowermost two layers, \ie~$L_{\ell}$ and~$L_{\ell-1}$, can be in a blocking coalition. If the vertex~$u$ in question is from the second last layer, then it will 
		only be in a blocking coalition $S$ if $S$ contains $u$, all the children of $u$ as well as the parent of $u.$ But then $S$ is not a blocking coalition for the children of $u$.
		For a leaf node $v$ to be in a blocking coalition, it will need to be with its parent $u$ but in a smaller coalition. But this means that $u$ is not in a blocking coalition. 
		We remove all vertices from coalitions only containing vertices from the last and second last layer and repeat the argument inductively.
		\end{proof}

\subsection*{Bakers and Millers: complete $k$-partite graphs}

\begin{theorem}
	Let $(N,\succsim)$ be a Bakers and Millers game with type space
	 $\types=\set{\type_1,\dots,\type_t}$ and~$\cstr=\set{S_1,\dots,S_m}$  a partition. Then, $\cstr$ is in the strict core if and only if for all types~$\type\in\types$ and all coalitions~$S,S'\in\pi$,  \[\frac{|S\cap\type|}{|S|}=\frac{|S'\cap\type|}{|S'|}\text.\]
\end{theorem}

\begin{proof}
	First assume that for all types~$\type\in\types$ and all coalitions~$S$ and~$S'$ in~$\pi$ we have  $\frac{|S\cap\type|}{|S|}=\frac{|S'\cap\type|}{|S'|}\text,$
but that a weakly blocking coalition~$T$ for~$\cstr$ exists. Then,  $\frac{|T\cap\type(j)|}{|T|}\le\frac{|\cstr[j]\cap\type(j)|}{|\cstr[j]|}$
for all $j\in T$, while
there is some~$i\in T$ with $\frac{|T\cap\type(i)|}{|T|}<\frac{|\cstr[i]\cap\type(i)|}{|\cstr[i]|}$. Consider this~$i$. 
Without loss of generality assume that $\type_1,\dots,\type_k$ are the types represented in~$T$, \ie those types~$\type$ with $j\in \type$ for some $j\in T$.
By assumption we have, for all~$j\in T$,
$\frac{|\cstr[j]\cap\type(j)|}{|\cstr[j]|}=\frac{|\cstr[i]\cap\type(j)|}{|\cstr[i]|}\text.$ Hence,
\[	\frac{|T\cap \baker|}{|T|}+\dots+\frac{|T\cap\type_k|}{|T|}<\frac{|\cstr[i]\cap\baker|}{|\cstr[i]|}+\dots+\frac{|\cstr[i]\cap\type_k|}{|\cstr[i]|}\text.
\]
Observe that both
\begin{align*}
	\frac{|T\cap \baker|}{|T|}+\dots+\frac{|T\cap\type_k|}{|T|}	=	1 
	\end{align*}
	\begin{align*}
	\text{\quad and \quad}
	\frac{|\cstr[i]\cap\baker|}{|\cstr[i]|}+\dots+\frac{|\cstr[i]\cap\type_k|}{|\cstr[i]|}\le 1\text.
\end{align*}
A contradiction follows.

For the other direction, assume that there are coalition $S,T\in\cstr$ and a type $\type\in\types$ such that
$	\fracc{S\cap\type}{S}>\fracc{T\cap\type}{T}\text.
$
Then, $S\cap\type\neq\emptyset$ and let $i\in S\cap\type$.
As
\[\frac{|S\cap\baker|}{|S|}+\dots+\frac{|S\cap\type_t|}{|S|}=\frac{|T\cap \baker|}{|T|}+\dots+\frac{|T\cap\type_t|}{|T|}\text,
\]
there is some type~$\type'\in\types$ such that
	$\fracc{S\cap\type'}{S}<\fracc{T\cap\type'}{T}\text.$
Accordingly, $T\cap\type'\neq\emptyset$. 

First consider the case in which both $S\cap\type'=\emptyset$ and $T\cap\type=\emptyset$. Without loss of generality, we may assume that $|S|\le|T|$. Observe that $|S|<|T\cup\set i|$. The coalition $T\cup\set{i}$ is weakly blocking, as
\[
	\fracc{(T\cup\set i)\cap\type}{T\cup\set i}
	=
	\fracc{\set i}{T\cup\set i}
	<
	\fracc{\set i}{S}
	\le
	\fracc{S\cap\type}{S}
\]
and
\[
	\fracc{(T\cup\set i)\cap\type''}{T\cup\set i}
	=
	\fracc{T\cap\type''}{T\cup\set i}
	\le
	\fracc{T\cap\type''}{T}
\]
for every type~$\type''$ distinct from~$\type$.
(The latter inequality is not strict, as $T\cap\type''$ may be empty.)

Finally, assume without loss of generality, that $T\cap\type\neq\emptyset$ and let $j\in T\cap\type$. Since~$S$ and~$T$ are distinct and both in~$\cstr$, also $i\neq j$. We show that the coalition $T'=(T\setminus\set j)\cup\set i$ is weakly blocking. Consider an arbitrary type~$\type''\in\types$. Observe that $|T|=|T'|$ and $|T\cap\type''|=|T'\cap\type''|$, whether $\type''=\type$ or not. Therefore,
$
	\fracc{T\cap\type''}{T}=\fracc{T'\cap\type''}{T'}.
$
Accordingly, every player~$k\in T\setminus\set{i,j}$ is indifferent between~$T$ and~$T'$. To conclude the proof, observe that
$	\fracc{T'\cap\type}{T'}=\fracc{\cstr[j]\cap\type}{T}.
$
 Hence,
\[
	\fracc{\cstr[i]\cap\type(i)}{S}
	=
	\fracc{S\cap\type}{S}
	>
	\fracc{T\cap\type}{T}
	=
	\fracc{\cstr[j]\cap\type}{T}
	=
	\fracc{T'\cap\type}{T'}\text,
\]
\ie~$T'\succ_i S$, as desired.
\end{proof}

\subsection*{Graphs with large girth}

We say that two vertices~$v$ and~$w$ have a \emph{neighbor in common in $(V,E)$} if either $\set{v,w}\in E$ or there is some $u\in V$ such that $\set{u,v},\set{u,w}\in E$.  This notion allows a useful characterization of graphs with girth of at least five.

\begin{lemma}\label{lemma:5girth}
	Let $(V,E)$ be a graph with $|V|\ge 3$. Then, $(V,E)$ has girth of at least five if and only if all distinct $v,w\in V$ have at most one neighbor in common.
\end{lemma}
\begin{proof}
	For the if direction, assume that $(V,E)$ contains a cycle of length three or four. In either case, it is easy to find vertices that have at least two neighbors in common. For the only-if direction, assume that there are $v,w\in V$ that have more than one neighbor in common. That is, either $\set{v,w}\in V$ and there is some $u\in V$ such that $\set{u,v},\set{u,w}\in V$ or 
	there are $u,u'\in V$ such that $\set{u,v},\set{u,w},\set{u',v},\set{u',w}\in V$. If the former, the graph has girth of at most three. If the latter, the graph's girth is at most four. 
\end{proof}

\begin{theorem}
	For simple and symmetric FHGs represented by graphs with girth at least five, the core is non-empty.
\end{theorem}

\begin{proof}	
	The reader is referred to Figure~\ref{fig:star_packing} for a graphical illustration of certain aspects of its proof.
	
We first introduce the more general notion of graph packing.   
Let $\mathcal{F}$ be a set of connected undirected graphs.
An \emph{$\mathcal{F}$-packing} of a graph~$G$ is a subgraph~$H$ of~$G$ such that each component of~$H$ is isomorphic to a member of~$\mathcal{F}$. The components of an $\mathcal F$-packing~$H$ can be seen as coalitions, and thus~$\mathcal F$-packings naturally induce a coalition partition, with each vertex not contained in a connected component forming a singleton coalition.
We will consider \emph{star-packings} of graphs, \ie $\mathcal F$-packings with $\mathcal{F}=\{S_2, S_3, S_4, \ldots, \}$ such that each $S_i$ is a star with~$i$ vertices. Each star~$S_i$ with $i>2$ has one center~$c$ and~$i-1$ leaves $\ell_1,\dots,\ell_{i-1}$. We view $S_2$ as having two centers and no leaves.

We will prove that a star packing that maximizes \emph{leximin welfare} is core stable.
Formally, with each star packing, denoted by~$\pi$, we associate an \emph{objective vector} $\vec x(\pi)=(x_1,\dots,x_{|V|})$ such that $x_i\le x_j$ if $1\le i\le j\le |V|$, and there is a bijection~$f\colon V\to \set{1,\dots,|V|}$ with $u_{v}(\pi)=x_{f(v)}$. Thus, in $\vec x(\pi)$ the vertices/players are ordered according to their value for~$\pi$ in ascending order.   We assume these objective vectors to be ordered lexicographically by $\ge$, \eg $\left(\frac{1}{2},\frac{1}{2},\frac{1}{2},\frac{1}{2}\right)\ge\left(0,\frac{1}{3},\frac{1}{3},\frac{2}{3}\right)$ but not vice versa. The goal is to compute a star packing that maximizes its objective vector. Intuitively, this balances the sizes of the stars in the star packing and does not leave vertices needlessly on their own. 
Clearly, star packings maximizing the objective are guaranteed to exist and in the remainder of the proof we argue that such star packings are in the core. 

Observe first that every graph~$(V,E)$ admits a star packing such that every vertex which is not isolated (\ie which has a neighbor) is contained in some star~$S_i$ for~$i \geq 2$. This can be seen by considering a spanning forest. Thus, every star packing of~$(V,E)$ that maximizes the objective vector must have this property.

Now, let~$\pi$ be a star packing of a graph~$(V,E)$ that maximizes the objective vector. 
For a contradiction, assume that there is a coalition~$S$ blocking~$\pi$. Then, $S$ contains no isolated vertices, as these always obtain utility~$0$ and can thus not be strictly better off in~$S$. Therefore,~$S$ consists entirely of vertices that are either centers or leaves of~$\pi$.
Also observe that, for any two leaves~$\ell,\ell'$ in~$\pi$ we have $\set{\ell,\ell'}\notin E$. For a contradiction assume that there were such leaves~$\ell,\ell'$. Then,~$\ell$ and~$\ell'$ must come from different centers, otherwise $(V,E)$ would contain a triangle. Moreover,  $\pi'=\set{\set{\ell,\ell'},\pi'_1\dots,\pi'_k}$, where $\pi_i'=\pi_i\setminus\set{\ell,\ell'}$, is a star packing for which the objective vector is larger than the one for~$\pi$, because all leaves in $\pi$ obtain weakly higher utility in $\pi'$, and $\ell$ and $\ell'$ obtain strictly higher utility.

Now three cases can be distinguished: (i) $S$ contains no centers of $\pi$, (ii) $S$ contains exactly one center of $\pi$ and (iii) $S$ contains more than one center of $\pi$.

If~(i),~$S$ only contains leaves of~$\pi$, between which we know there are no edges. Hence, every member of~$S$ has utility~$0$ and~$S$ cannot be blocking.

If~(ii), we show that~$\vec x(\pi)$ is not optimal.
 {Let~$S$ consists of one center~$c$ and $m$~leaves $\ell_1,\dots,\ell_m$ of $\pi$.} Since no leaves in~$\pi$ are neighbors, and $S$ does not contain isolated vertices, $S$ must be a star with $c$ as center and $\ell_1,\dots,\ell_m$ as leaves. Let~$\ell$ denote one of the leaves and~$c'$ the center of~$\pi$ such that $\ell\in\pi(c')$. Consider the partition~$\pi'$ such that 
\[
	\pi'(k) = 
	\begin{cases}
		\pi(c)\cup\set{\ell}	&	\text{if $k\in\pi(c)\cup\set{\ell}$, and}\\
		\pi(k)\setminus\set{\ell}& \text{otherwise.}
	\end{cases}
\]
We claim that $\vec x(\pi')>\vec x(\pi)$, contradicting our initial assumption. Observe that it suffices to prove that (a) $u_\ell(\pi')>u_\ell(\pi)$ and (b) $u_k(\pi')\ge u_\ell(\pi')$ for all~$k$ with $u_k(\pi')<u_k(\pi)$.

For (a), observe that if $u_c(\pi)< u_c(S)$ and $c$ is a center in both~$\pi$ and~$S$, then $\frac{|\pi(c)|-1}{|\pi(c)|}< \frac{|S|-1}{|S|}$. 
Moreover, $u_\ell(\pi)<u_\ell(S)$, \ie $\frac{1}{|\pi(\ell)|}<\frac{1}{|S|}$.
Accordingly, $|\pi(c)|<|S|<|\pi(\ell)|$. 
It follows that 
$|\pi'(\ell)|=|\pi(c)\cup\set{\ell}|\le S<|\pi(\ell)|$ 
and thus $u_\ell(\pi')>u_\ell(\pi)$.

For (b), let~$k$ be such that $u_k(\pi')<u_k(\pi)$. Then either $k=c'$ or $k\in\pi(c)\setminus\set{c}$. As~$c'$ is a center and~$\ell$ a leaf in~$\pi$, $c'$ still is a center in~$\pi'$. Hence, $u_{c'}(\pi')\ge\frac{1}{2}$. Moreover, $\ell$ is also a leaf in $\pi$ and thus $u_\ell(\pi')<\frac{1}{2}$, proving the case. Now assume that $k\in\pi(c)\setminus\set{c}$. Then, with~$k$ and~$\ell$ being both leaves in $\pi'(c)$, $u_k(\pi')=u_\ell({\pi'})$.

If (iii), assume that~$S$ contains at least two centers~$c$ and~$c'$ in~$\pi$. Then, $u_c(\pi)\ge\frac{1}{2}$ and $u_{c'}(\pi)\ge\frac{1}{2}$. Either $|S|=2k+2$ or $|S|=2k+3$ for some $k\ge 1$. As both $u_c(S)>\frac{1}{2}$ and $u_{c'}(S)>\frac{1}{2}$, also $|\set{i\in S\midd \set{c,i}\in E}|\ge k+2$ and $|\set{i\in S\midd \set{c',i}\in E}|\ge k+2$. It follows that~$c$ and~$c'$ must have at least two neighbors in common, contradicting Lemma~\ref{lemma:5girth}. This completes the proof.
\end{proof}

	\begin{figure}[t]
	\centering
	{
	\scalebox{.5}{
	\LARGE
	\begin{tikzpicture}
		\tikzstyle{pfeil}=[--,draw]
		\tikzstyle{pfeil}=[->,>=angle 60, shorten >=1pt,draw]
		\tikzstyle{onlytext}=[]
		\tikzstyle{node}=[circle,fill=gray!45]

		\draw 
			(-5,0)
			node[circle,fill=white,draw,inner sep=.4ex](c1){$c_1$}
			++(270:3) 
			+(180:1.5) node[circle,fill=white,draw,inner sep=.4ex](l1){$\ell_1$}
			+(  0:1.5) node[circle,fill=white,draw,inner sep=.4ex](l2){$\ell_2$}
		;                      
		\draw
			(0,0)
			node[circle,fill=white,draw,inner sep=.4ex](c2){$c_2$}
			++(270:3)                     
			node[circle,fill=white,draw,inner sep=.4ex](l4){$\ell_4$}
			+(180:1.7)                     
			node[circle,fill=white,draw,inner sep=.4ex](l3){$\ell_3$}
			+(0:1.7)                     
			node[circle,fill=white,draw,inner sep=.4ex](l5){$\ell_5$}
		;
		\draw 
			(5,0)
			node[circle,fill=white,draw,inner sep=.4ex](c3){$c_3$}
			++(270:3) node[circle,fill=white,draw,inner sep=.4ex](l7){$\ell_7$}
			+(180:1.7)                     
			node[circle,fill=white,draw,inner sep=.4ex](l6){$\ell_6$}
			+(0:1.7)                     
			node[circle,fill=white,draw,inner sep=.4ex](l8){$\ell_8$}
		;

\path[draw,very thick,-] (c1) -- (l1);
\path[draw,very thick,-] (c1) -- (l2);      

\path[draw,very thick,-] (c2) -- (l3);
\path[draw,very thick,-] (c2) -- (l4);
\path[draw,very thick,-] (c2) -- (l5);

\path[draw,very thick,-] (c3) -- (l6);
\path[draw,very thick,-] (c3) -- (l7);
\path[draw,very thick,-] (c3) -- (l8);

\path[draw,dashed,thick,-] (l2) -- (c2);
\path[draw,dashed,thick,-] (l4) -- (c3);
\path[draw,dashed,thick,-]	(c3) edge [bend right] (c1);
\path[draw,dashed,thick,-]	(l3) edge [bend right] (l8);

\end{tikzpicture}
	}}

	\caption{A graph with girth~$5$ and a star packing indicated by the solid edges. This star packing does not have an optimal objective vector: a better one would result if~$\ell_3$ and~$\ell_8$ were to form a star. Note that $\set{\ell_3,\ell_8}$ is a blocking coalition.}\label{fig:star_packing}
	\end{figure}
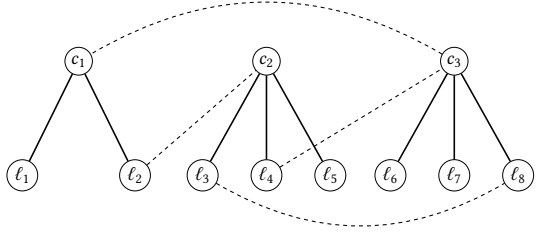

The procedure described in the proof above does not immediately yield a polynomial time algorithm that produces a core stable partition, since it is unclear whether a leximin star packing can be found in polynomial time. However, inspecting the proof further, we see that we in fact only need a \emph{local} optimum.

\begin{theorem}
	For simple and symmetric FHGs represented by graphs with girth at least five, an element of the core can be found in polynomial time.
\end{theorem}

\begin{proof}
	The existence proof above showed that if a given star packing $\pi$ is blocked by some coalition, then there exists a leximin-better star packing $\pi'$ that could be obtained from $\pi$ in one of the following two ways:
	\begin{enumerate}
		\item[(a)] two leaves $\ell, \ell'$ from different stars in $\pi$ with $\{\ell, \ell'\} \in E$ are removed from their respective coalitions and form the new star $\{\ell, \ell'\}$, or
		\item[(b)] a leaf $\ell$ is moved from one star to another.
	\end{enumerate}
	Our algorithm now proceeds as follows: start by producing some star packing of $G$ in which every non-isolated vertex is in a star (such a star packing can be found by considering a spanning forest of $G$). Then improve this star packing by using operations (a) and (b) if they lead to a leximin improvement, until no more such opportunities are available. The resulting star packing is in the core by the argument in the existence proof above.
	
	It remains to analyze the runtime of this algorithm. Clearly, the initial step and each improvement step can be executed in polynomial time, so we only need to establish that the algorithm terminates after a polynomial number of improvement steps.
	
	Define the following potential function for each star packing $\pi$:
	\[  \Phi(\pi) = \sum_{i\in V \text{center}} |V| + \sum_{i\in V \text{leaf}} |V| - |\pi(i)|.    \]
	Note that this potential function is integral, non-negative, and bounded above by $|V|^2$. We show that every time we perform (a) or (b), the potential strictly increases. This implies that at most $|V|^2$ improvement steps will be required.
	
	If we perform (a), then we convert two leaves $\ell$ and $\ell'$ into centers and thereby strictly increase their contribution to $\Phi$. We also decrease the sizes of the stars that $\ell$ and $\ell'$ were part of in $\pi$, which increases the contributions to $\Phi$ of the remaining leaves in those stars. Everyone else's contribution stays fixed.
	
	If we perform (b), using the notation of the previous proof, we move $\ell$ from $\pi(c')$ to $\pi(c)$. Since by case (ii)(a) of that proof we thereby increase the utility of $\ell$, the leaf $\ell$ has moved from a large star to a smaller star; in particular $|\pi(c')| \ge |\pi(c)| + 2$. After the move of $\ell$, the contributions to $\Phi$ of the leaves of $\pi(c')$ have each increased by 1, and the contributions of leaves of $\pi(c)$ have decreased by 1. Since there are more of the former than of the latter, this is an overall strict improvement.
\end{proof}

\end{document}